\documentclass[11pt]{article}
\usepackage[paper=a4paper,textwidth=162mm, textheight=238mm]{geometry}
\usepackage[T1]{fontenc}
\usepackage[utf8]{inputenc}
\usepackage{color}
\usepackage{latexsym}
\usepackage{amsmath, amssymb, amsfonts, amsthm}
\usepackage{dsfont}
\usepackage{pdfsync}
\usepackage{subfigure}
\usepackage[pdftex]{graphicx}
\usepackage[english]{babel}
\usepackage{multirow}
\usepackage{authblk}
\usepackage{pdfpages}
\newtheorem{Theorem}{Theorem}
\newtheorem{Definition}{Definition}
\newtheorem{Proposition}{Proposition}
\newtheorem{Assumption}{Assumption}
\newtheorem{Lemma}{Lemma}
\newtheorem{Corollary}{Corollary}
\newtheorem{Remark}{Remark}
\newtheorem{Example}{Example}
\makeatletter \@addtoreset{equation}{chapter}

\def \argsup{{\rm arg}\!\sup\limits}
\def \Sb{{\mathbf S}}
\def \Ac{{\cal A}}
\def \Cc{{\cal C}}
\def \Fc{{\cal F}}
\def \Nc{{\cal N}}
\def \Pc{{\cal P}}
\def \Uc{{\cal U}}
\def \E{{\mathbb E}}
\def \F{{\mathbb F}}
\def \Q{{\mathbb Q}}
\def \R{{\mathbb R}}
\def \P{{\mathbb P}}
\def \x{\times}
\def \demi{\frac{1}{2}}
\def \brace#1{\left\{ #1 \right\}}
\def \ind#1{\mathds{1}_{\brace{#1}}}
\def \cl {{\rm cl}}
\def \And{\;\mbox{ and }\;}
\def \Pas{\mathbb{P}-\mbox{a.s.}}
\def \eps{\varepsilon}
\def \vp{\varphi}
\def \Esp#1{\mathbb{E}\left[#1\right]}
\def \Pro#1{\mathbb{P}\left[{#1}\right]}

\title{Hedging Expected Losses on Derivatives \\ in Electricity Futures Markets}
\author[1]{Adrien Nguyen Huu\thanks{Corresponding author: adrien.nguyenhuu@gmail.com}}
\author[2]{Nadia Oudjane}
\affil[1]{{\it IMPA,  Rio de Janeiro, Brasil}}
\affil[2]{{\it EDF R\&D, Clamart, France}}

\begin{document}

\maketitle
\section{Introduction}
\label{sec:intro}

We propose in this contribution a method involving
numerical implementation for partially hedging
financial derivatives on
electricity futures contracts.
Electricity futures markets
present specific features.
As a non-storable commodity, 
electrical energy is delivered as
a power over time periods.
Similarly, futures contracts exchange
a present power price for delivery over a fixed future period
against the future quoted price on that period:
they refer explicitly to swap contracts, 
see \cite{benth2008stochastic}.
Electricity being non-storable,
arbitrage arguments do not hold,
preventing anyone to 
construct a term structure
via usual tools.
As a natural result of liquidity restriction and
uncertainty in the future,
a limited number of contracts are quoted
and the length of their delivery period increases with their time to maturity.
This phenomenon we call the \textit{granularity of the term structure}
implies that
the most desired and flexible contracts 
(weekly or monthly-period contracts) are 
only quoted a short time before
their maturity.
If one desires a fixed price for power delivered on a given distant month in the future, 
she shall get a contract covering a wider period, e.g.
a quarter or a year contract. 
However, the risk remains if one is
endowed with a derivative upon such a short period contract,
before this contract is even quoted.

This is the explored situation:
we consider here an agent endowed with a derivative
upon a non-yet quoted asset.
In practice, the unhedgeable risk
of such a position is managed
by deploying a cross-hedging
strategy with a 
correlated quoted asset, 
see \cite{grafenstein2003hedging}, \cite{eichhorn2005mean} and \cite{lindell2009strips}.
The risk cannot be completely eliminated
without prohibitive cost:
the market is incomplete and
the methodology 
requires a pricing criterion, 
as proposed in the above references.
In what follows,
we develop a partial hedging
procedure in order
to satisfy a loss constraint in expectation.
We aim at providing a ready-to-implement method to attain
such objectives.
This necessitates recent tools of
stochastic control \cite{bouchard2009stochastic}
and numerical approximations of 
coupled forward backward SDEs.
However, we relate constantly the obtained results to
well-known formulas and concepts,
so that the method is easy to assimilate.
We also provide strong and sufficient assumptions
in order to avoid involved proofs in
difficult cases.
We believe that the specific method proposed hereafter can 
be understood and applied without much effort, 
and be profitable to a numerous variety of hedging problems.

The approach we adopt has been
originated by F\"{o}llmer and Leukert
\cite{follmer1999quantile, follmer2000efficient}
but
we develop the problem
in the framework of control theory.
The minimal initial portfolio
value needed to satisfy
the constraint of expectation of losses
can be expressed as a value
function of a stochastic target problem.
The stochastic target approach 
has been introduced by Soner and Touzi
\cite{soner2002dynamic, soner2002stochastic}
to formulate the pricing problem in a control
fashion.
It has been extended to expectation criteria
by Bouchard, Elie and Touzi \cite{bouchard2009stochastic},
and applied for quantile hedging \cite{bouchard2009stochastic},
loss constraints in liquidation problems \cite{bouchard2010generalized},
or loss constraints with small transaction costs \cite{bouchard2013hedging}.
However,
the general approach is rather technical and
necessitates to solve a non-linear Partial Differential Equations (PDE),
when one has first proved a comparison 
theorem to resolve the uniqueness problem for the value function.

Here,
we use the application of \cite{bouchard2009stochastic}, 
where nice results can be provided
in complete market without appealing to comparison arguments.
However,
our initial problem cannot be
tackled directly with the stochastic target
formulation of \cite{bouchard2009stochastic} or \cite{moreau2010stochastic}.
The unhedgeable risk coming from
the apparition of the desired asset price
generates a non-trivial extension to the usual framework.
Instead,
we proceed in three steps in a backward fashion:
\begin{enumerate}
 \item We first formulate in Section \ref{sec:complete}
 the stochastic
 target problem in complete market,
 using an easy extension of the application in \cite{bouchard2009stochastic}.
 Our approach relies strongly on the 
 convexity of the loss function.
 By convex duality arguments, the
 expected loss target can be expressed
 in order to provide a formula which
 is very close to a usual risk-neutral pricing formulation.
 This allows to express the value function $V$
 at the precise moment of the apparition of the missing contract.
 \item To treat the random apparition of a non-yet traded asset,
 we use a face-lifting procedure that provides a
 new (intermediary) target for the period before
 quotation of the underlying asset. This is done
 in Section \ref{sec:incomplete},
 in coherence with the initial hedging criterion,
 and allows to retrieve a complete market setting
 where results of step 1 can be partially used.
 \item When the complete market setting cannot provide
 analytical formulas, as it happens after step 2, 
 we shall make use of
 numerical approximations. We propose such an algorithm
 in Section \ref{sec:numerics}, and illustrate 
 the efficiency of the method in Section \ref{sec:validation}.
\end{enumerate}
As it will be understood,
this method can be applied recursively by proceeding
repeatedly to steps 1 and 2 along the numerical
algorithm provided in step 3.
The remaining of the contribution
follows the above order, preluded
by the introduction of the problem
in section \ref{sec:model},
where we develop the archetypal situation encountered by
a financial agent on the market.

\section{Description of the control problem}
\label{sec:model}

Let $0<T<T^*<\infty$.
We consider an agent endowed with a financial option
with expiration date $T^\ast$, 
payoff $g$ on a futures covering
a monthly period.
However,
this asset is not yet quoted at initial time $t=0$,
and appears on the market at time $T$.
The month is covered by a futures
with delivery over a larger period,
e.g. a quarter futures.
We denote by $X:=(X_t)_{t\in [0,T^\ast]}$
the discounted price of this instrument,
avoiding the introduction of an interest rate
hereafter.
We assume it is available over the
whole period of interest $[0,T^{\ast}]$,
whereas the monthly-period futures 
upon which the agent has a position
is only available on the interval $[T, T^{\ast}]$.

The heart of our approach is to assume a
structural relation between the two instruments above,
namely that the return of the two assets are perfectly correlated.
More precisely, the monthly futures price is supposed to be given as the product, 
$\Lambda X_t$, of the quarterly futures price $X_t$ and 
a given  \textit{shaping factor} $\Lambda>0$,  
assumed to be a bounded random variable revealed at time $T$, completely
independent of the asset price $X$.
The boundness condition follows from a structural
relation between the two futures contracts,
assuming implicitly that the underlying spot price is non-negative.

We consider a probability space $(\Omega, \Fc, \P)$
supporting a Brownian motion $W$ and the variable $\Lambda$.
The filtration $\F$ is given by
$\Fc_t : = \sigma(W_s, \ 0\le s \le t)$ for
$t<T$, and by 
$\Fc_t := \sigma(W_s, \ 0\le s \le t )\wedge \sigma(\Lambda)$
for $t\ge T$.
\begin{Assumption}
\label{assumption:L}
We denote by $L$ the support of $\Lambda$,
which is supposed to be a bounded subset of $\R^+$,
and $\rho$ its probability measure on $L$.
\end{Assumption}

On the period $[0,T]$, 
the agent trades with the asset $X^{t,x}$,
which is assumed to be solution
to the SDE:
\begin{equation}
\label{eq:price dynamics}
 dX^{t,x}_s = \mu(s, X^{t,x}_s)ds + \sigma(s, X^{t,x}_s)dW_s , \quad \textrm{for}\ s\ge t , \And X^{t,x}_t=x\; .
\end{equation}
We assume the following for $X^{t,x}$ to be
well and uniquely defined.
\begin{Assumption}
\label{assumption:mu and sigma}
The functions
$
(\mu, \sigma) \ : \ [0,T^*]\x \R_+ \rightarrow \R \x \R_+^*
$
are assumed to be such that the following properties are verified:
\begin{enumerate}
	\item 
	$\mu$ and $\sigma$ are  uniformly Lipschitz
	and verify
	$
	|\mu(t,x)| + |\sigma(t,x)|\le K(1+|x|)
	$ for any $(t,x)\in [0,T^*]\x\R_+$.
	\item for any $x>0$, we have
	$X^{t,x}_s > 0$ $\Pas$ for all $s\in [0,T^*]$
	\item  $\sigma(t,x)>0 $ for any $(t,x)\in [0,T^*]\x\R^*_+$;
	\item if $x=0$ then $\sigma(t,x)=0$ for all $t\in [0,T^\ast]$;
	\item $\sigma$ is continuous in the time variable on $[0,T^\ast]\times \R_+$;
	\item $\mu$ and $\sigma$ are such that 
\begin{equation}
 \label{eq:Novikov}
 \theta(t,x):= \frac{\mu(t,x)}{\sigma(t,x)}<\infty \quad \text{ uniformly in }(t,x)\in [0,T^*]\x \R^*_+\; .
\end{equation}
\end{enumerate}
\end{Assumption}
Eq. \eqref{eq:Novikov} implies the so-called Novikov condition.
The submarket composed of only $(X_t)_{t\in [0,T^*]}$
is then a complete market associated to the Brownian subfiltration.
The filtration $\F$ relates to an incomplete market because of $\Lambda$, 
which is unknown on the interval $[0,T)$ 
and cannot be hedged by a self-financed admissible portfolio defined as follows.
In that manner,
the market can be labelled as semi-complete in the sense of Becherer \cite{becherer2001rational}.
\begin{Definition}
 \label{def:portfolio}
 An admissible self-financed portfolio is a $\F$-adapted process
 $Y^{t,x,y,\nu}$ 
 defined by $Y_t^{t,x,y,\nu}=y\ge -\kappa$ and
 \begin{equation}
\label{eq: portfolio process}
Y^{t,x,y,\nu}_s := y + \int_t^s \nu_u dX^{t,x}_u  \; , \quad
  s\in[t,  T^{\ast}]\; ,
\end{equation}
where $\nu\in \Uc_{t,y}$ denotes a strategy and
$\Uc_{t,y}$ the set of admissible strategies at time $t$,
which is the set of $\R$-valued $\F$-progressively measurable,
square-integrable processes
such that $Y_s^{t,x,y,\nu}\ge -\kappa \ \Pas$ for all  $s\in[t,  T^{\ast}]$ 
and some $\kappa\ge 0$ representing a finite credit line for 
the agent.
\end{Definition}

According to the asset model \eqref{eq:price dynamics},
it is strictly equivalent to 
consider a hedging portfolio with $X^{t,x}$ on $[0,T^{\ast}]$
or a switch at any time $r\in [T, T^{\ast}]$ for
the newly appeared asset $\Lambda X^{t,x}$:
\begin{equation*}
Y^{t,x,y,\nu}_s 
= y + \int_t^{s \wedge r} \nu_u dX^{t,x}_u + \int_{s \wedge r}^t \nu'_u d(\Lambda X^{t,x}_u) \; , 
r \ge T \And t\ge 0 \; ,
\end{equation*}
if $\nu'=\nu/\Lambda$ for $u\ge r$.
We thus assume that the agent trades with $X$ on $[0,T^*]$.

\begin{Assumption}
\label{assumption:g}
The final payoff of the option is given by
$g(\Lambda X_{T^{\ast}})$,
where the function $g$ is assumed to be Lipschitz continuous.
\end{Assumption}

As it is well-known in the literature,
the superhedging price of such an option can be prohibitive, 
even with Assumption \ref{assumption:L}
where $\Lambda$ is bounded.
To circumvent this problem, 
we propose to control the expected losses on partial hedging.
That means that the agent gives herself a 
threshold $p<0$ and a loss function $\ell$
to evaluate the loss
over her terminal position $g(\Lambda X^{t,x}_{T^{\ast}})-Y^{t,x,y,\nu}_{T^{\ast}}$.

\begin{Definition}
\label{assumption:l}
A loss function $\ell\ :\ \R_+ \rightarrow \R_+$ is assumed
to be continuous, strictly convex and strictly increasing on $\R_+$
with polynomial growth.
We normalize the function so that $\ell(0)=0$.
\end{Definition}

The agent's objective at time $t$ is to find 
the minimal value $y$ and a portfolio strategy $\nu\in \Uc_{t,y}$ such that
\begin{equation}
\label{eq: control loss}
\Esp{- \ell\left( \left(g\left(\Lambda X^{t,x}_{T^{\ast}}\right)-Y^{t,x,y,\nu}_{T^{\ast}}\right)^{+} \right)}\ge p \; .
\end{equation}
More generally, it will be useful 
to measure the deviation between the payoff and the hedging portfolio 
through a general map $\Psi: \R_+ \x \R \rightarrow \R_-$. 
The previous example corresponds to the specific case where 
\begin{equation}
\label{eq:Psi}
\Psi(x, y) = -\ell((g(x)-y)^{+}).
\end{equation}
Assumption \ref{assumption:g} and Definition \ref{assumption:l} imply that $\Psi$ defined above satisfies the following assumption:
\begin{Assumption}
\label{assumption:Psi}
The function $\Psi: \R_+ \x \R \rightarrow \R_-$ is assumed to be 
\begin{itemize}
	\item continuous with polynomial growth in $(x,y)$.
	\item  concave and increasing in $y$ on
	$
	\cl \brace{y\in [-\kappa, \infty) ~:~ \Psi(x,y)<0}
	$ for any $x\in \R_+$.
\end{itemize}
\end{Assumption}
Notice that $\Psi(x,y)=0$ in \eqref{eq:Psi} for any
$y\ge g(x)$, which makes $\Psi$
not invertible on the domain $\R$ for any fixed $x$.
Under Assumption \ref{assumption:Psi},
we can define the $y$ inverse $\Psi^{-1}(x,p)$ on $\R_+\x \R_-$
as a convex increasing function of $p$,
where 
\begin{equation}
 \label{eq: psi inverse}
 \Psi^{-1}(x,0) = \inf\brace{y\ge -\kappa ~:~ \Psi(x,y)=0} \; .
\end{equation}

As explained above,
the valuation approach of \eqref{eq: control loss} has been introduced
in \cite{follmer2000efficient, follmer2006convex}.
It can be written under the stochastic control form
of \cite{bouchard2009stochastic}
allowed by the Markovian framework \eqref{eq:price dynamics}-\eqref{eq: portfolio process}.
\begin{Definition}
 \label{def:stochastic control problem}
 Let $(t,x,p)\in \Sb:=[0,T^*]\x \R_+^*\x \R_-^*$ be given.
 Then we define the value function $V$ on $\Sb$ as
 \begin{equation}
 \label{eq: stochastic target}
V(t,x,p):= \inf \brace{ y\ge  -\kappa ~:~ 
\Esp{\Psi(\Lambda X^{t,x}_{T^{\ast}}, Y^{t,x,y,\nu}_{T^{\ast}})}\ge p \text{ for }\nu \in \Uc_{t,y} } \; .
\end{equation}
\end{Definition}
Advanced technicalities and details
on the general setting for the stochastic
target problem with controlled loss can be found 
in \cite{bouchard2009stochastic} and \cite{moreau2010stochastic}.
In particular,
we introduce the value function $V$ only on 
the open domain of $(x,p)$, and avoid
treating the specific case $x=0$ or $p=0$.
Furthermore, the function
$V$ is implicitely bounded by the superhedging price of $g(\Lambda X^{t,x}_{T^*})$
as we will see in the next section.
Notice that the expectation in~(\ref{eq: stochastic target}) 
involves an integration w.r.t. the law of $\Lambda$ 
for $t\in[0,T)$, before the shaping factor $\Lambda$ is revealed. 
Thus,
the present problem is not standard due to the presence of $\Lambda$:
the filtration $\F$ is not only due to the Brownian motion,
and dynamic programming arguments of \cite{bouchard2009stochastic}
do not apply directly.
The approach we undertake
is to separate the complete
and incomplete market intervals
in a piecewise problem.

\begin{Example}
 \label{ex:lower partial moment}
 A particular example of loss function
 we will use in sections \ref{sec:numerics} and \ref{sec:validation}
 is the special case of 
lower partial moment 
\begin{equation}
\label{eq:l}
\ell(x):= x^k \ind{x\ge 0}/ k \ ,\quad \textrm{with}\quad 
k> 1\ .
\end{equation}
In the case of $k=1$, (which is not considered here) 
we obtain a criterion close to the expected shortfall,
independently studied in \cite{cvitanic2000minimizing},
whereas allowing $k=0$ allows to
retrieve precisely the quantile hedging problem \cite{bouchard2009stochastic},
although $\ell$ is not a loss function as in Definition \ref{assumption:l} in that case.
The case $k=2$ gets closer to the quadratic hedging
(or mean-variance  hedging) objective, but with
the advantage of considering only losses, and not gains.
Notice that $\ell$ as in equation \eqref{eq:l} for $k>1$
follows Definition \ref{assumption:l} and 
that Assumption \ref{assumption:Psi} holds with 
\eqref{eq:Psi} if Assumption \ref{assumption:g} holds.
\end{Example}

We finish this section with additional notations.
In the sequel,
we will denote $\Sb:=\Sb_1 \cup \Sb_2$
where $\Sb_1:=[0,T]\x \R_+^*\x\R_-^*$ and
$\Sb_2:=[T,T^*]\x \R_+^*\x\R_-^*$
are the two domains of the value function.
For any real valued function $\vp$, defined on $\Sb$, we denote
by $\vp_t$ or $\vp_x$ the
partial derivatives with respect to
$t$ or $x$. Partial derivatives of other variables,
or second order partial derivatives are written 
in the same manner.

\section{Solution in complete market}
\label{sec:complete}

\subsection{General solution and risk-neutral expectation}
\label{sec:general complete}

On the interval $[T,T^{\ast}]$,
the variable $\Lambda$
is known and the asset $\Lambda X$ is tradable.
We assume that $\Lambda$ takes the value $\lambda\in L$.
On this interval,
$\lambda$ can be seen as a coefficient affecting the
loss function $\Psi$, and
the problem \eqref{eq: stochastic target}
follows the standard formulation of \cite{bouchard2009stochastic}:
the filtration is generated by the paths of the Brownian motion.
In section \ref{sec:incomplete},
we reduce the incomplete market setting of period $[0,T)$
to complete market problem of the form \eqref{eq: stochastic target} on $[T,T^*]$, 
in a similar Brownian framework, with a function $\Psi$ 
which is no more derived from a loss function as in~\eqref{eq:Psi}, 
motivating the introduction of the general function $\Psi$.

Consequently and without loss of generality,
we temporarily assume $L=\{1\}$, 
and omit $\lambda$ in the notation.
On $[0,T^*]$ then,
we can work on problem \eqref{eq: stochastic target}
in the Brownian filtration and apply most results
of \cite{bouchard2009stochastic}.
The filtration $\F$ for $t\ge T$
in that case is given by $\Fc_t:=\sigma(W_s ~:~ T\le s \le t)$
on a complete probability space $(\Omega, \Fc, \P)$
with $\Pro{\Lambda=1}=1$.
\begin{Proposition}
 \label{prop:solution dual}
 Under Assumptions \ref{assumption:mu and sigma} and \ref{assumption:Psi}, the function $V$ 
is given on $\Sb$ by
\begin{equation}
\label{eq: dual formulation}
V(t,x,p) = \E^{\Q_t} \left[ \Psi^{-1}\left( X^{t,x}_{T^*}, 
J\left(X^{t,x}_{T^*}, Q^{t,q}_{T^*}\right) \right) \right]
\end{equation}
where
$J(x,q):=\argsup \brace{pq - \Psi^{-1}(x, p)~:~p \le 0}$
and
\begin{equation}
 \label{eq: risk neutral diffusion}
  X^{t,x}_s = x+ \int_t^s \sigma(u, X^{t,x}_u)  dW^{\Q_t}_u \And 
 Q^{t,q}_s = q + \int_t^s Q^{t,q}_u \theta(u, X^{t,x}_u)  dW^{\Q_t}_u
 \;,
\end{equation}
with $W^{\Q_t}$ being the Brownian motion under
the probability $\Q_t$, the latter being
defined by ${d\P}_t/{d\Q_{t}} = Q^{t,1}$
and used for expectation \eqref{eq: dual formulation}.
Finally $q$ is given such that
$
\E\left[ J\left(X^{t,x}_{T^*}, Q^{t,q}_{T^*}\right) \right]=p
$.
\end{Proposition}
To obtain such a result,
we proceed in several steps.
We first use Proposition 3.1 in \cite{bouchard2009stochastic}
to express
problem \eqref{eq: stochastic target}
in the standard form of \cite{soner2002dynamic}, 
see Lemma \ref{lem: stochastic target form} below.
We then apply the Geometric Dynamic Programming Principle (GDP)
Principle of \cite{soner2002dynamic}
in order to obtain a
PDE characterization of $V$ in the
viscosity sense.
These results are recalled in Appendix.
We then use properties
of the convex conjugate of $V$ to
obtain $V$ as the risk-neutral expectation \eqref{eq: dual formulation} 
of Proposition \ref{prop:solution in complete market}.

\begin{Lemma}
 \label{lem: stochastic target form}
 Let $P^{t,p,\alpha}$ be a $\F$-adapted stochastic process defined by
 \begin{equation}
  \label{eq: process P}
  P^{t,p,\alpha}_s = p + \int_t^s \alpha_u P^{t,p,\alpha}_u  dW_u \;, \quad 0\le t \le s \le T^* \; ,
 \end{equation}
where $\alpha$ is a $\F$-progressively measurable
process taking values in $\R$.
Let us denote $\Ac_{t,p}$ the set of such processes such that
$P^{t,p,\alpha}$ and $\alpha P^{t,p,\alpha}$ are square-integrable processes.
Let Assumptions \ref{assumption:mu and sigma} and \ref{assumption:Psi} hold.
Then on $\Sb$
\begin{equation}
\label{eq: stochastic equivalence}
 V(t,x,p) = \inf \brace{y\ge -\kappa\ :\ Y_{T^*}^{t,x,y,\nu} \ge \Psi^{-1}(X^{t,x}_{T^*}, P_{T^*}^{t,p,\alpha})
 \text{ for } (\nu, \alpha)\in  \Uc_{t,y}\x\Ac_{t,p} } \; .
\end{equation}

Moreover, for a given triplet $(t,y,p)$,
if there exists $\nu\in \Uc_{t,y}$ and a
$\F$-progressively measurable process $\alpha$ taking values in $\R$
such that 
\begin{equation}
 \label{eq:condition extention}
 Y_{T^*}^{t,x,y,\nu} \ge \Psi^{-1}(X^{t,x}_{T^*}, P_{T^*}^{t,p,\alpha}) \ \Pas \; ,
\end{equation}
then there exists $\alpha'\in \Ac_{t,p}$ such that
$Y_{T^*}^{t,x,y,\nu} \ge \Psi^{-1}(X^{t,x}_{T^*}, P_{T^*}^{t,p,\alpha'})$ $\Pas$
\end{Lemma}

\begin{proof}
Under Assumption \ref{assumption:Psi},
$\Psi^{-1}(x,p)$ is well defined with \eqref{eq: psi inverse}.
According to the polynomial growth of $\Psi$
combined with Assumption \ref{assumption:mu and sigma},
the stochastic integral representation theorem can be applied.
The first result \eqref{eq: stochastic equivalence}
is then a reformulation of Proposition 3.1 in \cite{bouchard2009stochastic}.

The second result echoes Assumption 4 and Remark 6 in 
\cite{bouchard2010generalized},
as the statement is missing in \cite{bouchard2009stochastic}.
According to \eqref{eq: process P},
for $p< 0$, the process $P^{t,p,\alpha}$ is a negative
local martingale, and therefore a  bounded submartingale.
Thus, 
$\Esp{\Psi(X^{t,x}_{T^*}, Y^{t,x,y,\nu}_{T^*})}\ge p$.
According to the growth condition of $\Psi$,
the martingale representation theorem implies the existence of
a square-integrable martingale $P^{t,p,\alpha'}$,
with $P^{t,p,\alpha'}_t =p$ and according to \eqref{eq:condition extention},
$$
\Esp{\Psi(X^{t,x}_{T^*}, Y^{t,x,y,\nu}_{T^*}) - P_{T^*}^{t,p,\alpha'}}
= \Esp{\Psi(X^{t,x}_{T^*}, Y^{t,x,y,\nu}_{T^*})}-p 
\ge 0 \; .
$$
Since $\Psi(x,y)\in \R_-$ for any $(x,y)\in \R_+\x \R$,
we can choose $P^{t,p,\alpha'}$ 
to follow dynamics (\ref{eq: process P}).
This implies that 
$\alpha'$ is a real-valued $(\Fc_t)$-progressively measurable process
such that $(\alpha' P^{t,p,\alpha'})\in L^2([0,T^*]\x \Omega)$,
and $\alpha'\in \Ac_{t,p}$.
\end{proof}

We can now turn to the proof of 
Proposition \ref{prop:solution in complete market}
by using the GDP recalled in Appendix.
The proof follows closely developments of 
section 4 in \cite{bouchard2009stochastic},
and is given for the sake of clarity,
as for the introduction of important objects
for section \ref{sec:expectation}.

\begin{proof}{\it (Proposition \ref{prop:solution dual})}
We divide the proof in three steps 

{\bf 1.} {\it We introduce conjugate and local test functions.}
Let $V_*$ be the lower semi-continuous
version of $V$, as defined in Appendix.
According to Assumption \ref{assumption:Psi},
dynamics \eqref{eq: portfolio process}
and definition \eqref{eq: stochastic target},
$V$ and $V_*$ are increasing functions of $p$ on $\R_-$.
The boundedness of $V$ implies also the finitness of $V_*$.
For $(t,x,q)\in [0,T^*]\x\R_+^*\x\R_+^*$,
we introduce the convex conjugate of $V_*$ in $p<0$, i.e.,
$$
\tilde V(t,x,q) :=\sup\brace{pq - V_*(t,x,p)~:~ p\le 0}\; .
$$
The map $q \mapsto \tilde V(.,q)$ is convex 
and upper-semi-continuous on $\R_+^*$.

Let $\tilde \vp$ be a smooth function with bounded derivatives, 
such that $(t_0, x_0, q_0)\in [0, T^*)\x\R_+^*\x\R_+^*$ 
is a local maximizer 
of $\tilde V-\tilde \vp$ with $(\tilde V-\tilde \vp) (t_0, x_0, q_0)=0$.
The map $q\mapsto \tilde\vp(.,q)$ is convex.
Without loss of generality,
we can assume that $\tilde \vp$ is strictly convex
with quadratic growth in $q$,
see the proof in Section 4 of \cite{bouchard2009stochastic}.

The convex conjugate of $\tilde \vp$ with respect to $q$ is
a strictly convex function of $p$ defined by
$
 \vp(t,x,p):=\sup\brace{qp - \tilde \vp(t,x,q) ~:~ q\ge 0}
$.
We can then properly define
the map $(t,x,q)\mapsto (\vp_p)^{-1}(t,x,q)$
on $[0,T^*]\x\R_+^*\x\R_+^*$,
where the inverse is taken in the $p$ variable.
According to the definition of $\tilde V$
and the quadratic growth of $\tilde \vp$,
there exists $p_0\le 0$ such that for the fixed $q_0$,
$$
 p_0 q_0 - V_* (t_0, x_0, p_0)  = \tilde V(t_0,x_0, q_0) = \tilde \vp(t_0, x_0, q_0) 
 = \sup_{p\le 0}\brace{pq_0 - \vp(t_0,x_0,p)}
$$ 
which, by taking the left and right sides of the above equation,
implies that $(t_0, x_0, p_0)$ is a local minimizer of $V_* - \vp$
and $(V_*-\vp)(t_0, x_0, p_0)=0$.
The first order condition in the definition of $\vp$
implies that $p_0=(\vp_p)^{-1}(t_0,x_0, q_0)$.

{\bf 2.} {\it We prove that $\tilde V$ is subsolution to a linear PDE.}
In our case, the control $\nu$ takes
unbounded values by definition of $\Uc_{t,y}$. 
It implies, together with Assumption \ref{assumption:mu and sigma}, that
\begin{equation}
 \label{eq: kernel nonempty}
 \Nc_0 (t,x, p, d_x, d_p):= \brace{(u,a)\in \R^2 ~:~\left|\sigma(t,x)\left(u- d_x\right) - ap d_p \right|=0 }
\ne \emptyset
\end{equation}
for any $(t,x,p,d_x, d_p)\in [0, T^*]\x\R_+\x \R_-\x\R^2$.
This holds in particular for the set
$$\Nc_0(t_0, x_0, p_0, \vp_x(t_0, x_0, p_0), \vp_p(t_0, x_0, p_0))$$
which is then composed of elements of the form
$
( (\vp_x
+ ap\vp_p/\sigma ) (t_0,x_0,p_0)  , a)$
for $a\in \R$.
According to Theorem \ref{th 2.1} in Appendix
and changing $\nu$ for its new expression,
$\vp$ in $(t_0, x_0, p_0)$ verifies
\begin{equation}\label{eq:PDEv de preuve}
-\vp_t - \demi \sigma^2(t_0, x_0) \vp_{xx} - 
\inf_{a\in \R}\brace{\demi (a p_0)^2 \vp_{pp}-a p_0 (\theta(t_0, x_0) \vp_p - \sigma\vp_{xp}) }\ge 0\ .
\end{equation}
Since $\vp_{pp} (t_0, x_0, p_0)>0$,
the infimum in the above equation is reached for
\begin{equation}
 \label{eq:meilleur controle}
 a :=- \left(\frac{\sigma  \vp_{xp} - \theta \vp_p}{p_0  \vp_{pp}} \right) (t_0,x_0, p_0)\in \R\; ,
\end{equation}
which can be plugged back into \eqref{eq:PDEv de preuve}
to obtain a new inequality at $(t_0, x_0, p_0)$:
\begin{equation}
 \label{eq: nonlinear for dual}
 - \vp_t -\frac{1}{2}\sigma^2(t_0, x_0)\vp_{xx} 
+\demi (\vp_{pp})^{-1}\big (\theta(t_0, x_0)  \vp_p - \sigma(t_0, x_0)   \vp_{xp}\big)^2  \geq 0 \; .
\end{equation}
Recall that $p_0=(\vp_p)^{-1}(t_0, x_0, q_0)$.
According to the Fenchel-Moreau theorem,
$\tilde \vp$ is its own biconjugate,
$\tilde \vp(t,x,q) = \sup \brace{pq - \vp(t,x,p) ~:~ p\le 0}$,
and 
$\tilde \vp(t_0,x_0,q_0) = p_0 q_0 - \vp(t_0, x_0, p_0)$.
By differentiating in $p$,
we get
$\vp_p(t_0, x_0, p_0) = q_0$.
It follows by differentiating again that
at point $(t_0, x_0, p_0)$ we have
the following correspondence:
\begin{equation}
 \label{eq:dual relation}
 \left(\vp_t,\vp_x, \vp_{xx},\vp_{pp},\vp_{xp} \right)=
 \left(-\tilde \vp_t,-\tilde \vp_x, -\tilde \vp_{xx} + \frac{\tilde\vp_{xq}^2}{\tilde\vp_{qq}} ,\frac{1}{\tilde \vp_{qq}},
 - \frac{\tilde \vp_{xq}}{\tilde \vp_{qq}} \right)\; .
\end{equation}
Plugging \eqref{eq:dual relation} into \eqref{eq: nonlinear for dual},
we get that $\tilde \vp$
satisfies  at $(t_0, x_0, q_0)$
\begin{equation}
\label{eq:lineaire}
-\tilde \vp_t
- \demi\left(\sigma^2 \tilde \vp_{xx} + |\theta|^2 q_0^2 \tilde \vp_{qq}+2 \mu \tilde \vp_{xq} \right)(t_0, x_0, q_0)
\le  0 \; .
\end{equation}
By arbitrariness of $(t_0, x_0, q_0)\in [0,T^*)\x\R_+^*\x\R_+^*$,
this implies that $\tilde V$ is a
subsolution of \eqref{eq:lineaire}
on $[0,T^*)\x\R_+^*\x\R_+^*$.
The terminal condition is given by 
the definition of $\tilde V$ and 
Theorem \ref{th 2.1} in Appendix: 
\begin{equation}
 \label{eq:terminal condition}
 \tilde V(T^*,x,q) = \sup_{p\le0}\brace{pq - V_* (T^*,x,p)}
=\sup_{p\le 0}\brace{pq - \Psi^{-1}(x,p)} \; .
\end{equation}

{\bf 3.}
{\it We compare $V$ to a conditional expectation.}
Let $\bar{V}$ be the function  defined by 
$
\bar{V}(t,x,q) = \E^{\Q_t}\left[ \tilde V(T^*, X^{t,x}_{T^*}, Q^{t,q}_{T^*}) \right]
$
for 
$(t,x,q)\in [0, T^*]\x (0,\infty)^2$, 
with dynamics for $s\in [0, T^*]$ given by
$$
 X^{t,x}_s = x+ \int_t^s \sigma(u, X^{t,x}_u)  dW^{\Q_t}_u \And 
 Q^{t,q}_s = q + \int_t^s Q^{t,q}_u \theta(u, X^{t,x}_u)  dW^{\Q_t}_u
 \;,
$$
where $\Q_t$ is a $\P$-equivalent measure 
such that ${d\P}/{d\Q_{t}} = Q^{t,1}$.
According to the Feynman-Kac formula,
$\bar{V}$ is a supersolution to equation \eqref{eq:lineaire},
and thus $\bar{V}\ge \tilde V$.
According to Assumption \ref{assumption:Psi},
$p \mapsto \Psi^{-1}(.,p)$ is convex increasing on $\R_-$.
Thus, for sufficiently large values of $q$,
$
J(x,q) := \arg \sup\brace{pq - \Psi^{-1}(x,p) ~:~p\le 0}
$ is well-defined and can take any value in $\R_-$.
By the implicit function theorem,
there exists a function
$q$ of $(t,x,p)$ such that
$
\E^{\Q_{t}}\left[ Q^{t,1}_{T^*} J\left(X^{t,x}_{T^*}, Q^{t,q(t,x,p)}_{T^*}\right) \right]=p
$.
Therefore,
$$
\begin{array}{lcl}
V(t,x,p) 
&\ge& V_*(t,x,p) 
 \ge \sup\brace{qp - \bar{V}(t,x,q) ~:~ q\ge 0 } \\
 
&\ge& pq(t,x,p) - \E^{\Q_{t}}\left[ \tilde V \left(T^*, X^{t,x}_{T^*}, Q^{t,q(t,x,p)}_{T^*}\right) \right]\\

 &\ge& q(t,x,p) \left(p - \E^{\Q_{t}}\left[Q^{t,1}_{T^*} J\left(X^{t,x}_{T^*}, Q^{t,q(t,x,p)}_{T^*}\right) \right]\right) \\

& & \hspace{1cm}
  + \E^{\Q_{x}}\left[\Psi^{-1}\left(X^{t,x}_{T^*},J\left(X^{t,x}_{T^*}, Q^{t,q(t,x,p)}_{T^*}\right)\right) \right]\\
  
 &\ge& \E^{\Q_{t}}\left[ \Psi^{-1}\left(X^{t,x}_{T^*},J\left(X^{t,x}_{T^*}, Q^{t,q(t,x,p)}_{T^*}\right) \right)\right]
=: y(t,x,p) \; .
\end{array}
$$
By the martingale representation theorem,
there exists $\nu\in \Uc_{t,y}$ such that
$$
Y^{t,y(t,x,p),\nu}_{T^*} = \Psi^{-1}\left(X^{t,x}_{T^*},J\left(X^{t,x}_{T^*}, Q^{t,q(t,x,p)}_{T^*} \right) \right)
$$
which implies that for $p\le 0$
$$
\begin{array}{ll}
\Esp{\Psi\left(X^{t,x}_{T^*},Y^{t,y(t,x,p),\nu}_{T^*}\right)} 
& \ge \Esp{J\left(X^{t,x}_{T^*}, Q^{t,q(t,x,p)}_{T^*}\right)}
 = \E^{\Q_{t}}\left[Q^{t,1}_{T^*} J\left(X^{t,x}_{T^*}, Q^{t,q(t,x,p)}_{T^*}\right) \right] \\
& \ge p
\end{array}
$$
and therefore, by definition of the value function
$y(t,x,p) \ge V(t,x,p)$
which provides the equality and \eqref{eq: dual formulation} 
for $(t,x,p)\in \Sb$.
\end{proof}

\subsection{Application to the interval $[T,T^*]$}
\label{sec:interval complete}

We now come back to $[T, T^*]$, 
where the \textit{shaping factor} $\Lambda$, 
is assumed to take the known value $\lambda$
at time $T$.
To highlight the effect of
$\lambda$ as a given state parameter for $t\ge T$,
we denote by $V(t,x,p,\lambda)$ the 
value function defined similarly as~\eqref{eq: stochastic target} where $\Psi$ is given by~\eqref{eq:Psi}.

\begin{Definition}
 \label{def:stochastic target with lambda}
 Let $(t,x,p,\lambda)\in \Sb_2\x L$. We can define the value function on $\Sb_2 \x L$ as
 \begin{equation}
 \label{eq: stochastic target bis}
V(t,x,p,\lambda):= \inf \{ y\ge  -\kappa ~:~ 
\Esp{\ell\left(g\left(\lambda X^{t,x}_{T^{\ast}}\right)-Y^{t,x,y,\nu}_{T^{\ast}}\right)^+}\le -p \ \text{ for }\nu \in \Uc_{t,y} \} \; ,
\end{equation}
\end{Definition}
Notice that if $X$ is an exponential process,
then we can explicitely change
$V(t,x,p,\lambda)$ for $V(t,\lambda x, p, 1)$,
recalling the assumption $\Pro{\Lambda=1}=1$ of the previous section.
Let us recall the standard pricing 
concepts in complete market.
Under Assumption \ref{assumption:mu and sigma}, 
we can define 
$\Q$ the $\P$-equivalent martingale measure defined by
\begin{equation}
\label{eq: EMM}
\left.\frac{d\Q}{d\P}\right|_{\Fc_t} = \exp \brace{-\int_T^{t} \theta(s,X_s) dW_s - 
\demi \int_T^{t} |\theta(s,X_s)|^2 ds} \;, t\ge T \; .
\end{equation}
In the present setting,
$\Q$ is the unique risk-neutral measure 
with the drifted Brownian motion
$W_t^{\Q} = W_t + \int_T^t \theta(s, X_s)ds$,
and
we can provide a unique no-arbitrage price for
$g(\lambda X_{T^{\ast}})$.
\begin{Definition}
 \label{def:black-scholes price}
 We define for $(t,x,\lambda)\in [T, T^*]\x\R_+^*\x L$ the function
 \begin{equation}
\label{eq: black-scholes price}
C(t,x, \lambda):=\E^{\Q}[g(\lambda X^{t,x}_{T^{\ast}})]\quad  \text{ for } \ t\ge T \; .
\end{equation}
\end{Definition}
According to Assumption \ref{assumption:mu and sigma}
and Assumption \ref{assumption:g}, 
we can apply Proposition~6.2, in~\cite{janson}, 
which implies that for any $\lambda \in L$, $(t,x)\mapsto C(t,x,\lambda)$ 
is Lipschitz continuous in the spatial variable with the same Lipschitz constant $K$ as $g$. 
Moreover, $C(\cdot,\cdot,\lambda)$ is the unique classical solution 
of polynomial growth to the
Black-Scholes equation
%
\begin{equation}
 \label{eq:black-scholes equation}
 -C_t - \demi \sigma^2(t,x) C_{xx}=0 \ \text{ on }\ [T, T^*)\x (0,+\infty)
\end{equation}
with terminal condition $C(T^*,x, \lambda)=g(\lambda x)$.
Now Proposition \ref{prop:solution dual}
applies to provide the following when we use a loss function
as in Definition \ref{assumption:l}.
\begin{Corollary}
 \label{prop:solution in complete market}
 Let Assumptions \ref{assumption:mu and sigma} and \ref{assumption:g} hold. 
 Then $V$ is given on $\Sb_2 \x L$ by 
\begin{equation}
\label{eq:explicit solution complete market +}
V(t,x,p, \lambda) = \E^{\Q} [g(\lambda X^{t,x}_{T^*}) - \ell^{-1}(-P_{T^*}^{t,p}) ]
\end{equation}
where $P^{t,p}_{T^*}$ is a $\Fc_{T^{\ast}}$-measurable random variable
defined by
\begin{equation}
\label{eq:diffusion of P}
P^{t,p}_{T^*} = j \left(q \exp\left( \int_t^{T^{\ast}} \theta(s,X^{t,x}_s)dW^{\Q}_s - 
\demi  \int_t^{T^{\ast}} \theta^2(s,X^{t,x}_s)ds \right)\right)
\end{equation}
with $j(q):= -((\ell^{-1})')^{-1}(q)$
and
$q$ in \eqref{eq:diffusion of P} such that
$
\E^{\Q} \left[ P^{t,p}_{T^*} \right] =  p \; .
$
Moreover, $V$ is convex and increasing in $p$ and is $\Cc^{1,2}$ in $(t,x)$ on $[T,T^*]\x\R_+^*$.
\end{Corollary}
The solution \eqref{eq:explicit solution complete market +}-\eqref{eq:diffusion of P}
can be explicitly computed in simple cases,
see section \ref{sec:numerics} below.
Notice that $V$ is bounded on $[T, T^{\ast}]\x \R_+^*\x\R_-^*\x L$ by
$C$.
Since $\ell$ is convex increasing on $\R_+$,
$-\ell^{-1}$ is convex. 
A look at \eqref{eq:explicit solution complete market +}-\eqref{eq:diffusion of P}
then convinces that $V$ is continuous in $p$
on $\R_-^*$, but
Proposition 3.3 in \cite{bouchard2009stochastic}
can be used in our setting
to assert that $V$ is also convex in $p$.
According to what was said about the function $C$,
$V$ is $\Cc^{1,2}$ in $(t,x)$ on $[T,T^*]\x\R_+^*$.
\begin{Remark}
\label{rem:partial hedging}
It is noticeable that the
value function \eqref{eq:explicit solution complete market +} 
is composed of the
Black-Scholes price of the claim $g(\lambda X^{t,x}_{T^*})$ minus
a term that corresponds
to a penalty in the dual expression
of the acceptance set
if $\ell$ is a risk measure, see \cite{follmer2006convex}.
The hedging strategy is to be modified in consequence.
First note that
\eqref{eq:explicit solution complete market +}
is a conditional expectation of
a function of two Markov processes $(X^{t,x}, Q^{t,q})$,
for $q$ well chosen,
so that we can write $Y^{t,V(t,x,p), \nu}_s := y(X^{t,x}_s,Q^{t,q}_s) := V(s,X^{t,x}_s,P^{t,p}_s)$.
According to Corollary \ref{prop:solution in complete market},
$y$ is a regular function, and $X^{t,x}, Q^{t,q}$ are martingales under
the probability $\Q$, as well as $Y^{t,V(t,x,p), \nu}$. Therefore
It\^o formula provides
\begin{equation*}
dY^{t,V(t,x,p), \nu}_s
=y_x (X^{t,x}_s,Q^{t,q}_s)dX^{t,x}_s
+  y_q (X^{t,x}_s,Q^{t,q}_s)dQ^{t,q}_s \; .
\end{equation*}
Now expressing $dW^{\Q}_s$ with respect to $dX^{t,x}_s$,
we obtain 
\begin{equation}
\label{eq: strategy complete}
dY^{t,V(t,x,p), \nu}_s= 
\left(y_x(X^{t,x}_s,Q^{t,q}_s) + \mu(s, X^{t,x}_s) Q^{t,q}_s y_q(X^{t,x}_s,Q^{t,q}_s) \right) dX^{t,x}_s
\end{equation}
which allows to deduce the optimal dynamic strategy $\nu$.
\end{Remark}

Corollary \ref{prop:solution in complete market}
retrieves solutions of \cite{follmer2000efficient}.
The original problem in the latter
is to minimize the expected loss
given an initial portfolio value, but
the authors prove that our version of the problem
is equivalent.
The stochastic target problem of Definition \ref{def:stochastic control problem}
is actually linked to an
optimal control problem in a similar way,
as noticed in the introduction of \cite{bouchard2009stochastic} and
developped in \cite{bouchard2011optimal}.
We can introduce the value function giving the minimal loss 
that can be achieved at time $T^\ast$, 
with  at time $t$, 
the initial capital $y$,  
$X^{t,x}_t=x$ and the shaping factor $\Lambda=\lambda$.
\begin{Definition}
 \label{def:value function u}
 For $(t,x,y,\lambda)\in [T,T^*]\x\R_+^*\x\R\x L$ we define
 \begin{equation}
\label{eq:optimal control problem}
U(t,x,y,\lambda) := \sup \brace{ \Esp{-\ell\left(\left(g\left(\lambda X^{t,x}_{T^*}\right)-Y_{T^{\ast}}^{t,y,\nu}\right)^+\right)} 
~:~ \nu\in \Uc_{t,y} }\ .
\end{equation}
\end{Definition}
This corresponds to the problem of finding the best reachable threshold $p$
if the initial portfolio value is given by $y$ at time $t$.
This result is of great use in the
forthcoming resolution of the problem before $T$,
by the following connection with $V$.
\begin{Lemma}
 \label{lem:equivalence}
 For $(t,x,y,\lambda)\in [T,T^*]\x\R_+^*\x[-\kappa, \infty)\x L$
 we define 
 $$V^{-1}(t,x,y,\lambda) 
:= \sup \brace{p \le 0 ~:~ V(t,x,p, \lambda)\leq y } \; .$$
 Then we have
 \begin{equation}
\label{eq:equivalence result}
U(t,x,y,\lambda) = V^{-1}(t,x,y,\lambda)  \; .
\end{equation}
Moreover the function $U$ is a concave increasing function of $y$ bounded by $0$ from above.
\end{Lemma}
Equality \eqref{eq:equivalence result} is a direct application
of Lemma 2.1 in \cite{bouchard2011optimal},
whereas the properties of $U$ fall from
the properties of $V$ and Definition \ref{def:value function u}.

\section{Solution in incomplete market}
\label{sec:incomplete}

We now turn to the solution of the problem
before $T$, i.e., when $\Lambda$
is unkown.
We first provide the
face-lifting procedure that
allows to reduce
the new situation to the one handled in Section \ref{sec:complete}.
This procedure can be done
according to two model paradigms:
\begin{enumerate}
	\item  the law $\rho$ of $\Lambda$ is supposed to be known;
	\item only the support $L$ of the law of $\Lambda$ is supposed to be known.
\end{enumerate}
The first approach is a \textit{probabilistic approach} 
with a prior distribution,
whereas the second one is a \textit{robust approach},
in connection with
robust finance with parameter uncertainty,
see \cite{bouchard2012stochastic}
for a control theory version.
Finally,
the numerical complexity of the face-lifting procedure
pushes us to give up
explicit formulas for a
numerical approach.
We thus modify results of
Section \ref{sec:complete}
to provide a convenient formulation
of the problem
to be numerically approximated in Section \ref{sec:numerics}.

\subsection{Faced-lifted intermediary condition with prior on the shaping factor}
\label{sec:face-lifting}

When $t< T$,
the problem \eqref{eq: stochastic target} cannot be treated with the methodology developped in~\cite{bouchard2009stochastic}.
To solve the problem,
we are guided by the following argument.
Considering $(t,y)\in [0, T)\x [-\kappa, \infty)$ 
and a strategy $\nu\in \Uc_{t, y}$,
we arrive at time $T$ to the wealth $Y^{t,x,y,\nu}_{T}\ge -\kappa$, 
at the apparition of the exogenous risk factor $\Lambda$.
Assume that the agent wants to
control the expected level of risk $p<0$, at time $T$, by using the portfolio $Y^{t,x,y,\nu}_{T}$.
It is obviously not possible with certainty
if $\Lambda$ takes a value $\lambda$ such that
$
Y_{T}^{t,y,\nu} < V(T, X^{t,x}_{T}, p, \lambda) \; .
$
However, 
the optimal strategy after $T$ consists
in optimizing the portfolio
by trying to achieve the optimal expected level of loss
$U(T, X^{t,x}_{T}, Y_{T}^{t,y,\nu}, \lambda)$ given by
Definition \ref{def:value function u}.
In the complete market setting,
this achievement is possible.

\begin{Lemma}
 \label{lem:selection}
 Under Assumptions \ref{assumption:mu and sigma} and \ref{assumption:Psi},
 there exists a map $(x,y, \lambda)\mapsto \nu(x,y,\lambda)\in \Uc_{T, y}$
 on $\R_+^*\x[-\kappa, \infty)\x L$ such that
 \begin{equation}
 \label{eq:measurable assumption}
 \Esp{ \Psi\left(\lambda X^{T,x}_{T^*}, Y^{T,y,\nu(x,y, \lambda)}_{T^*}\right)} \ge U(T,x,y,\lambda) \; .
\end{equation}
\end{Lemma}

\begin{proof}
 Fix $(x,y,\lambda)\in \R_+^*\x[-\kappa, \infty)\x L$.
 Let $p:=U(T,x,y,\lambda)$.
 According to Lemma \ref{lem:equivalence},
 $y\ge V(T, x, p, \lambda)$.
 Following
 Corollary \ref{prop:solution in complete market} and Remark \ref{rem:partial hedging},
 and since $\Psi$ is increasing in $y$,
 there exists $\nu$ such that
 $
 \Esp{\Psi\left(\lambda X^{T,x}_{T^*}, Y^{T,y,\nu}_{T^*}\right)} \ge p
 $ for $(x,y,\lambda)$ arbitarily fixed. 
\end{proof}
Then the expected loss at $T$ resulting from 
this strategy is averaged among the the realizations of $\Lambda$
for fixed $X^{t,x}_{T}$ and $Y_{T}^{t,y,\nu}$.
The expectation is thus done also with regard to $\Lambda$ just before $T$.
\begin{Definition}
 \label{def:face-lifted}
 We define $\Xi : \R_+\times [-\kappa,\infty)\rightarrow \R_-$ by
 \begin{equation}
\label{eq:face lifted terminal condition}
\Xi (x,y) := \int_{L} U(T, x, y, r) \rho(dr) \; .
\end{equation}
\end{Definition}
The above function $\Xi$ represents the expected optimal
level of loss the agent can reach if
she attains the wealth $y$ at time $T$ with a state $X^{T,x}_{T}=x$. 
\begin{Lemma}
 \label{lem:properties of Xi}
 The function $\Xi$ takes non-positive values, 
 is $\Cc^{2}$ in $x\in \R_+^*$
 and concave increasing in $y$.
\end{Lemma}
This is a direct consequence of Lemma \ref{lem:equivalence} and corollary~\ref{prop:solution in complete market}.
This property ensures that parts of Assumption \ref{assumption:Psi} hold
in the new following problem, in order to apply Theorem \ref{th 2.1} of Appendix.
In particular, Lemma \ref{lem:properties of Xi} allows to define
a terminal condition with the generalized inverse $\Xi^{-1}$.
\begin{Definition}
 \label{def:new problem}
 For $(t,x,p)\in \Sb_1$, we define
 \begin{equation}
\label{eq:new stochastic target problem}
\bar V(t,x,p):= \inf \brace{y\ge -\kappa ~:~ \Esp{\Xi(X^{t,x}_{T}, Y^{t,x,y,\nu}_{T})} \ge p 
\text{ for }\nu \in \Uc_{t,y} } \; .
\end{equation}
\end{Definition}
We now prove 
that 
this new problem coincides
with the one of Definition \ref{def:stochastic control problem} on $S_1$.


\begin{Proposition}
\label{prop:equivalence T}
Let Assumptions \ref{assumption:mu and sigma} and \ref{assumption:Psi}
hold.
Then
$\bar V(t,x,p) = V(t,x,p)$ on $\Sb_1$.
\end{Proposition}

\begin{proof}
{\bf 1.}
Fix $(t,x,p)\in \Sb_1$ and
take $y> V(t,x,p)$.
Then by definition there exists
$\nu\in \Uc_{t,y}$ such that
$\Esp{\Psi\left(X^{t,x}_{T^*}, Y^{t,x,y,\nu}_{T^*}\right)}\ge p$.
Since $t<T$, the control can be written
$\nu_t  = \nu_t \ind{t\in [0,T)} + \nu_t(\Lambda)\ind{t\in [T, T^*]}$
where 
$\nu_t(.)$ follows from the canonical construction of $\F$:
it is a measurable map from $L$ to
the set of
square integrable control processes on $[T, T^*]$ which are
adapted to the Brownian filtration $\sigma(W_s, \ T \le s \le .)$.
From the flow property
of Markov processes $X^{t,x}$ and $Y^{t,x,y,\nu}$ (see \cite{soner2002dynamic})
and the tower property of expectation,
\begin{equation}
 \label{eq:flow property}
 \begin{array}{ll}
 \Esp{\Psi\left(X^{t,x}_{T^*}, Y^{t,x,y,\nu}_{T^*}\right)} = \\
 \displaystyle \Esp{ \int_L  
 \E \left[ \Psi\left(X^{T,X^{t,x}_T}_{T^*}, Y^{T,X^{t,x}_{T},Y^{t,x,y,\nu}_T,\nu(\lambda)}_{T^*}\right)
 \Big| \left(X^{t,x}_T,Y^{t,x,y,\nu}_T, \lambda\right) \right] \rho(d\lambda)  } \;.
 \end{array} 
\end{equation}
By taking the supremum over all possible maps $\nu(\lambda)$ and 
by Definition \ref{def:value function u}, 
we obtain
\begin{equation}
 \label{eq:inequality face}
 p\leq \E \left[ \Psi\left(X^{T,X^{t,x}_T}_{T^*}, Y^{T,Y^{t,x,y,\nu}_T,\nu(\lambda)}_{T^*}\right)
 \Big| \left(X^{t,x}_T,Y^{t,x,y,\nu}_T, \lambda\right) \right]\le U\left(T, X^{t,x}_T,Y^{t,x,y,\nu}_T, \lambda\right)\; .
\end{equation}
By integrating in $\lambda$ over $L$,
and then take the expectation,
we immediately get that $y\ge \bar V(t,x,p)$.
By arbitrariness of $y$, $V(t,x,p)\ge \bar V(t,x,p)$.

{\bf 2.}
Take $y> \bar V(t,x,p)$.
There exists a
control $\nu\in \Uc_{t,y}$ on $[t, T]$
such that
$$
\Esp{\int_L U\left(T, X^{t,x}_{T}, Y^{t,x,y,\nu}_{T}, \lambda\right)\rho(d\lambda)}\ge p \; .
$$
Now 
Lemma \ref{lem:selection}
allows to assert the existence of a
control $\nu^*(\lambda)$ on $[T, T^*]$ such that 
for any $\lambda\in L$
\begin{equation*}
 \Esp{\Psi\left(X^{T, X^{t,x}_T}_{T^*}, Y^{T, Y^{t,x,y,\nu}_T, \nu^*(\lambda)}_{T^*} \right) } 
 = U\left(T, X^{t,x}_{T}, Y^{t,x,y,\nu}_{T}, \lambda\right) \; .
\end{equation*}
By chosing the new admissible control $\nu'\in \Uc_{t,y}$
defined by the concatenation
$\nu'_t  = \nu_t \ind{t\in [0,T)} + \nu^*_t(\Lambda)\ind{t\in [T, T^*]}$,
we obtain $y\ge V(t,x,p)$ for an arbitrary $y> \bar V(t,x,p)$.
We thus have equality of $V$ and $\bar V$ on $\Sb_1$.
\end{proof}
In the present context when
the law of $\Lambda$ is known,
the face-lifting procedure of \eqref{eq:face lifted terminal condition}
allows to retrieve a stochastic
target problem in expectation of \cite{bouchard2009stochastic}
on the interval $[0,T]$.
However,
it is improbable that 
the terminal condition \eqref{eq:face lifted terminal condition}
is explicit for non-trivial models.

\subsection{Variation to a robust approach}
\label{sec:robust}

Due to the lack of data,
or some non stability of the problem,
it might be interesting to adopt an approach
where the agent wants to control the expected level of loss
without assuming the law $\rho$ on $L$.
Under Assumption \ref{assumption:L},
the robust approach is easy to undertake and
is the following.
Since the law of $\Lambda$ is not known,
we have to consider the worst case scenario, i.e.,
the supremum of \eqref{eq:face lifted terminal condition}
over a set of probability measures on $L$.
\begin{Definition}
 \label{def:robust approach}
 Let $\Pc(L)$ be the set of all probability measures
over $L$ including singular measures. Then
we define for any $(x,y)\in \R_+^*\x[-\kappa, \infty)$ the function
\begin{equation}
 \label{eq:optimal over singular}
 \xi(x,y) := \sup\limits_{\rho\in \Pc(L)} \Xi(x,y) \;.
\end{equation}
\end{Definition}
It is straightforward that
by considering singular measures in $\Pc(L)$
we shall have
\begin{equation}
 \label{eq:robust terminal condition xi}
 \xi (x,y) = \max\limits_{\lambda \in L} U(T,x,y,\lambda)\;,
\end{equation}
which is finite for all $(x,y)\in \R_+^*\x[-\kappa, \infty)$
following Lemma \ref{lem:equivalence}.
The reasoning of the previous section
can be applied here with the new intermediary condition
$\Esp{\xi(X^{t,x}_{T}, Y^{t,x,y,\nu}_{T})}\ge p$.
We thus introduce the other stochastic target problem
with controlled loss if $t<T$:
\begin{equation}
 \label{eq:robust stochastic target}
 \bar V^R(t,x,p):= \inf \brace{y>-\kappa ~:~ \Esp{\xi(X^{t,x}_{T}, Y^{t,x,y,\nu}_{T})} \ge p 
\text{ for }\nu \in \Uc_{t,y}} \; .
\end{equation}
The monotonicity of $U$ with respect to
$\lambda$,
provided by some additional assumption on $g$,
leads to a direct
solution to Eq. \eqref{eq:robust terminal condition xi}.
In general, the resolution of problem \eqref{eq:robust stochastic target}
is similar to \eqref{eq:new stochastic target problem}
with a terminal condition that is more tractable.
In the sequel, we thus focus on problem \eqref{eq:new stochastic target problem}.

\subsection{From the control problem to an expectation formulation}
\label{sec:expectation}

In this section, 
we want to emphasize formally the link between 
the solution of the nonlinear PDE~\eqref{eq:PDEv de preuve} 
(in the proof of Proposition \ref{prop:solution dual}) 
and  a simple conditional expectation, 
in  \textit{sufficiently regular settings}. 
This idea will be used to propose 
a numerical scheme to approximate the solution of 
our partial hedging problem. 
Assume fo starting that the nonlinear PDE~\eqref{eq:PDEv de preuve} 
has a classical solution $\bar V$ on  $\Sb_1$ 
such that it also verifies
\begin{equation}
\label{eq:HJB system}
\left \{
\begin{array}{l}
\bar V_t +\dfrac{\sigma^2(t,x)}{2}\bar V_{xx}
+ a^\ast p\left( \sigma(t,x)\bar V_{xp}
- \theta(t,x) \bar V_{p}\right)
+\dfrac{1}{2} (a^{\ast} p)^2\bar V_{pp}
= 0\ , \ \text{for } t<T\\
\bar V(T,x,p)=  \Xi^{-1}(x,p)	\ ,
\end{array}
\right.
\end{equation}
where we recall the specific form of the control \eqref{eq:meilleur controle}:
\begin{equation}
\label{eq:solution of a}
a^\ast :=\left(p\bar V_{pp}\right)^{-1} 
\left( \theta(t,x)\bar V_p-\sigma(t,x)\bar V_{xp}\right)\; ,
\end{equation}
and the terminal condition is
\begin{equation}
\label{eq:inverse of terminal condition}
\Xi^{-1}(x,p):= \inf\brace{y\ge -\kappa ~:~ \Xi(x,y)\ge p } \; .
\end{equation}
In particular if
$a^*$ in \eqref{eq:solution of a}
is well defined (by strict convexity of $\bar V$ in $p$)
and corresponds to the optimal value in \eqref{eq:PDEv de preuve},
then fomulations \eqref{eq:PDEv de preuve} and \eqref{eq:HJB system}
are equivalent.
Now let us assume that the function 
$\hat{V}$  such that 
\begin{equation}
\label{eq:expectation formula}
\hat{V}(t,x,p)=\E^{\Q}[\Xi^{-1}(X^{t,x}_{T},P^{t,p,\alpha^*}_{T})]\;,\quad 0\le t\le T \;,
\end{equation}
is well defined, 
with dynamics under $\Q$ given by
\begin{equation}
 \label{eq:FK dynamics}
 X^{t,x}_s = x+ \int_t^s \sigma(u, X^{t,x}_u) dW^{\Q}_u 
 \ ;  \ 
 P_s^{t,p,\alpha^*} = p + \int_t^s P^{t,p,\alpha^*}_{u}\alpha^\ast_u \big (dW^{\Q}_u-\theta(u,X^{t,x}_u)\big ) du
\end{equation}
and where the feedback control $\alpha^*$ is defined for $s\in [t,T]$ by
\begin{equation}
\label{eq:optimal control alpha}
\alpha^\ast_s = 
\left(P^{t,p,\alpha^*}_s \bar V_{ pp}\right)^{-1}  
\left( \theta(s,X^{t,x}_s)\bar V_{ p} -\sigma(s, X^{t,x}_s)\bar V_{ xp}
\right)\left(s, X^{t,x}_s, P^{t,p,\alpha^*}_s\right)  \; .
\end{equation}
Assume moreover that $\hat V$ is sufficiently regular to be a classical solution to the related linear Feynman-Kac PDE
\begin{equation}
\label{eq:FK}
\left \{
\begin{array}{l}
 \varphi_t +\dfrac{\sigma^2(t,x)}{2} \varphi_{xx}
+ a^\ast p\left( \sigma(t,x)\varphi_{xp}
- \theta(t,x) \varphi _{p}\right)
+\dfrac{1}{2} (a^{\ast} p)^2\varphi _{pp}
= 0\ , \ \text{for } t<T\\
\varphi(T,x,p)=  \Xi^{-1}(x,p)	\ ,
\end{array}
\right.
\end{equation}
where we recall the specific form of the control \eqref{eq:solution of a}
given as a function of $\bar V$ and its derivatives. 
Now observe that $\bar V$ is also a classical solution to this linear PDE. 
Hence, 
if the classical solution to~\eqref{eq:FK} is unique 
then we can conclude that $\hat V=\bar V$. 
In the following section, 
we propose a numerical scheme to solve~\eqref{eq:HJB system} which relies, 
in some sense on this relation between the non linear PDE ~\eqref{eq:HJB system} 
and the conditional expectation \eqref{eq:expectation formula}
with dynamic \eqref{eq:FK dynamics}-\eqref{eq:optimal control alpha}.  

\section{Numerical approximation scheme}
\label{sec:numerics}

In the present section,
we propose a numerical algorithm
dedicated to the specific problem
\eqref{eq:expectation formula}-\eqref{eq:FK dynamics}-\eqref{eq:optimal control alpha}.
The problem is specific for two reasons.
First, the terminal condition at $T$
is given by $\Xi^{-1}$ wich
may be unexplicit and has to be numerically studied.
Second, the optimal control
is explicitely given by \eqref{eq:optimal control alpha}
so that the non-linear operator
can be replaced by a proper expectation approximation. 
At this stage, the algorithm is presented as the result of 
a series of standard approximations. 
However, we do not provide any analysis of 
the approximation error induced by this algorithm 
so that it can only be considered as an heuristic. 
Nevertheless, some numerical simulations are provided 
in the next section and emphasize 
the practical interest of such numerical scheme. 

\subsection{Specification and hedging in complete market}
\label{sec:specification}

In this section,
we will retrieve all the regularity assumptions
by specifying the model.
The question of the generalization
of the presented procedure is
a matter that is not treated in this paper.
We assume that $X$ is described by a geometrical
Brownian motion:
\begin{equation}
\label{eq:para}
\mu (t, x) =\mu x\ ,\quad \textrm{and}\quad \sigma (t,x)=\sigma x\ ,
\end{equation}
with $(\mu, \sigma)\in \R \x \R_+^*$.
The loss function is given as in Example \ref{ex:lower partial moment},
with $k>1$.
For the partial lower moment function
$\ell(x) = x^k \ind{x\ge 0} /k$ for $k> 1$,
Corollary \ref{prop:solution in complete market}
has an explicit solution, given for $t\ge T$ by
\begin{align}
\label{eq:solution complete partial moment}
 V(t,x,p, \lambda) &= C(t, x, \lambda) - 
(-kp)^{1/k} \E^\Q \left[\exp\brace{\frac{1}{2(k-1)} \int_t^{T^{\ast}} |\theta(u,X^{t,x}_u)|^2 du }\right] \notag \\
 &= C(t, x, \lambda) - 
(-kp)^{1/k} \exp\brace{\frac{\theta^2}{2(k-1)}(T^*-t)}\;,\ \textrm{where}\ \theta=\frac{\mu}{\sigma}\ ,
\end{align}
where in this precise case,
$C(t,x,\lambda)$ is given by the Black-Scholes price of 
the option with payoff
$x\mapsto g(\lambda x)$, as in Definition \ref{def:black-scholes price}. 
Following section \ref{sec:complete}, $(t,x)\mapsto C(t,x,\lambda)\in \Cc^{1,2}([T, T^*)\x \R_+^*)$ for any $\lambda\in L$,
so that according to \eqref{eq:solution complete partial moment}
all the required partial derivatives of
$V$ exist.
Note also that $V$ is strictly convex in $p$ since $k>1$.

Consequently,
we can explicit the strategy to 
hedge the expected loss constraint.
If $a^*$ is given by \eqref{eq:solution of a},
then 
\begin{equation}
 \label{eq: optimal strategy}
 \nu^* (t,x,p)= \left(V_x + \frac{a^* p V_p}{x\sigma}\right) (t,x,p)\; .
\end{equation}
All the required derivatives are given by
$$
\left\{
\begin{array}{lll}
V_t (t,x,p,\lambda) \hspace{-0.2cm}
&=& C_t (t,x, \lambda) 
+ \frac{\theta^2}{2(k-1)}\exp\brace{\frac{\theta^2}{2(k-1)} (T-t) }\\ 
V_x (t,x,p, \lambda) \hspace{-0.2cm}
&=& C_x (t,x, \lambda)\\ 
V_{xx}(t,x,p, \lambda) \hspace{-0.2cm}
&=& C_{xx}(t,x, \lambda)\\ 
V_p (t,x,p, \lambda) \hspace{-0.2cm}
&=& \exp\brace{\frac{1-k}{k}\log (-kp) + \frac{\theta^2}{2(k-1)} (T-t) } \; .
\end{array} 
\right.
$$
As anticipated in Remark \ref{rem:partial hedging},
the strategy consists in hedging the claim $g(\lambda X^{t,x}_{T^*})$
plus a correcting term corresponding in 
hedging the constraint $P^{t,p,\alpha^*}_{T^*}$.

\subsection{The intermediary target}
\label{sec:terminal condition}

In order to initiate the numerical procedure for $0\le t\le T$,
we need to compute the intermediary condition
and its partial derivatives intervening in \eqref{eq:HJB system}.
According to \eqref{eq:equivalence result} and \eqref{eq:solution complete partial moment},
\begin{equation}
 \label{eq:solution inverse partial moment}
 U(t,x,y,\lambda)=\frac{-1}{k} (C(t,x,\lambda) - y)^k 
\exp\brace{-\frac{k \theta^2}{2(k-1)}(T^*-t) } \ind{C(t,x,\lambda) \ge y} \;,
\end{equation}
which provides the value of $\Xi(x,y)$ by integration according to
the law $\rho$ of $\Lambda$:
\begin{equation}
 \label{eq:xi numerical}
 \Xi(x,y) = \frac{-1}{k} \exp\brace{-\frac{k\theta^2}{2(k-1)}(T^* - T)} 
 \int_L (C(T,x, \lambda)-y)^k \ind{C(T,x,\lambda) \ge y} \rho(d\lambda) \; .
\end{equation}
A numerical computation of the integral in \eqref{eq:xi numerical}
can be proceeded via numerical integration
or Monte-Carlo expectation w.r.t. the law $\rho$. 
This in turn allows to obtain the desired function
$\Xi^{-1}$, since the latter is monotonous in $p$.

For 
fixed $(x,p)\in \R_+^*\x \R_-^*$,
define 
$$
M:=\brace{\lambda \in L ~:~ C(T,x,\lambda)-\Xi^{-1}(x,p)\ge 0 }\; .$$
Let us introduce four real-valued functions
$(f_{k-1}, \tilde f_{k-1}, f_{k-2}, \tilde f_{k-2})$ of
$(x,p)\in \R_+^*\x \R_-^*$ defined by
\begin{equation}
\label{eq:j}
\left \{
\begin{array}{lll}
f_n(x,p)&=&\int_{M} \big(C(T,x,\lambda)-\Xi^{-1}(x,p)\big )^{n}\rho(d\lambda)\\
\tilde f_n(x,p)&=& \int_{M} \lambda C_x (T,x,\lambda)
\big (C (T,x,\lambda)-\Xi^{-1}(x,p)\big )^{n}\rho(d\lambda)\; ,
\end{array}
\right .
\end{equation}
for $n=k-1,k-2$, recalling that $k$ denotes the exponent parameter determining the loss function~\eqref{eq:l}. 
Then by a straightforward calculus we derive the partial derivative of $\Xi^{-1}$ as follows 
\begin{equation}
 \label{eq:terminal derivatives}
 \left\{
\begin{array}{lcl}
\Xi^{-1}_x (x,p)		&=& \frac{\tilde f_{k-1}}{f_{k-1}}(x,p)\\
\Xi^{-1}_p (x,p)		&=& \exp \brace{\frac{\theta^2 k}{2(k-1)}(T^*-T)}   \dfrac{1}{f_{k-1}}(x,p)\\
\Xi^{-1}_{pp}(x,p)		&=& \left[(k-1) \frac{f_{k-2}}{f_{k-1}}\left(\Xi^{-1}_p \right)^2 \right](x,p)\\
\Xi^{-1}_{xp} (x,p)		&=& \left[(k-1)\Xi^{-1}_p \frac{\big(\tilde f_{k-1}f_{k-2}-\tilde f_{k-2}f_{k-1}\big)}
{\big (f_{k-1}\big )^2}
\right](x,p)\ .
\end{array}
\right.
\end{equation}
The functions $(f_{k-1}, \tilde f_{k-1}, f_{k-2}, \tilde f_{k-2})$
can be computed numerically with the same methods
as $\Xi^{-1}$.
Having these derivatives, it is thus possible
to obtain the values of controls at time $T$, 
$(\nu^*(T,x,p), a^*(T,x,p))$, given by \eqref{eq: optimal strategy} and \eqref{eq:solution of a}.

\subsection{Discrete time approximation and regression scheme}
\label{sec:scheme}

The approximation scheme that is proposed here is first based on 
a time discretization of the forward-backward dynamic determined by 
the system~\eqref{eq:expectation formula}-\eqref{eq:FK dynamics}-\eqref{eq:optimal control alpha}.

\subsubsection{Time discretization}
Let us define a deterministic time grid
$\pi:=\{0=t_0<\ldots<t_N := T\}$ with regular mesh
$|t_{i+1}-t_i|= T/N=: \Delta t$. We consider in this paragraph a dicrete time approximation of the process $(X^{0,x_0}, P^{0,p_0,\alpha^*})$ solution of~\eqref{eq:FK dynamics} with initial condition at time $0$, $(X^{0,x_0}_0, P^{0,p_0,\alpha^*}_0)=(x_0,p_0)$. 
The process $X^{0,x_0}$ possesses an 
exact discretization at times $(t_i)_{i=0..N}$
which is denoted $(X_i)_{i=0..N}$,
with increments of the Brownian motion
given by
\begin{equation}
\label{eq:epsilon}
W_{t_{i+1}} - W_{t_i}:=\sqrt{\Delta t }\epsilon_i\ ,
\end{equation}
 $(\epsilon_i)_{i=0,\cdots ,N-1}$ being a
sequence of i.i.d. centered and standard Gaussian random variables.
We introduce the sequence of random variables $(P_i^{a^*})_{i=0\cdots N}$ obtained by taking the exponential of the  Euler approximation of $\log (P^{0,p_0,\alpha^*})$  on the mesh $\pi$. 
Then, we can approximate the solution of~\eqref{eq:FK dynamics}, 
at the mesh instants $\pi$ by the Markov chain $(X_i,P^{a^*}_i)_{i=0,\cdots ,N}$ 
satisfying the following dynamic for $i=0, \ldots, N-1$:
\begin{equation}
\left \{\begin{array}{l}
\label{eq:XP}
 X_{i+1} = X_i \exp\brace{\sigma \sqrt{\Delta t}\epsilon_i - (\sigma^2 \Delta t)/2 }\\
 P^{a^*}_{i+1} = P^{a^*}_{i}\exp \brace{-a^*_i(X_i, P^{a^*}_i)\Big (\big (\theta+\frac{1}{2}a^*_i(X_i, P^{a^*}_i)\big )\Delta t +  \sqrt{\Delta t}\epsilon_i\Big ) }

 \end{array}
 \right .
\end{equation}
with the initial condition, 
$X_0=x_0$ and $P^{a^*}_0=p_0$ and where 
at each time step $i$,  $a^*_i$ is actually 
the function given by \eqref{eq:solution of a} at time $t_i$ 
\begin{equation}
 \label{eq:discrete optimal control}
 a_i^* (x,p): = a^*(t_i,x,p)=\frac{\theta p \bar V_p(t_i,x,p) - \sigma x p \bar V_{xp}(t_i,x,p)}{pp \bar V_{pp}(t_i,x,p)} 
\end{equation}
In the sequel, we will denote $X_{i+1}^{i,x}$ and $P_{i+1}^{i,x,p,a_i}$ 
the random variables satisfying equation \eqref{eq:XP} 
with $X_i=x$, $P^{a^*}_i=p$ and the function $a^*_i=a_i$. 

\subsubsection{Piecewise constant approximation of $a_i^*$ and tangent process formula}
Assume that at the dicrete time $t_{i}$, 
for $i\in \{0,\cdots ,N-1\}$, 
a piecewise constant approximation of $a_{i}^*$ is available, 
such that for any positive reals $x$ and $p$, 
we have the approximation $\hat{a}_i$ defined as follows 
\begin{equation}
\label{eq:approxa}
\hat a_{i}(x,p)=\sum_{r=1}^R a_{i,r}\mathbf{1}_{C_{i,r}}(x,p)\ ,
\end{equation}
where  $(C_{i,r})_{r=1,\cdots R}$ is a partition of 
$\R_+^*\x \R_-^*$ and $(a_{i,r})_{r=1,\cdots R}$ 
is a sequence of reals. 
By the expectation formula \eqref{eq:expectation formula}, 
we obtain that the solution of problem \eqref{eq:new stochastic target problem} 
and equivalently \eqref{eq: stochastic target} satisfies 
the following backward dynamic for $i\in \{0,\cdots,N-1\}$ 
$$
\hat V(t_i,x,p)=\E[\hat V(t_{i+1},X^{t_i,x}_{t_{i+1}},P^{t_i,p,\alpha^*}_{t_{i+1}})]\quad 
\textrm{for}\ (x,p)\in \R_+^*\x \R_-^*\ .
$$
Then, 
one can approximate $\hat V$ at the discrete instants of the mesh $\pi$, 
by injecting in the above formula two approximations consisting in:
\begin{enumerate}
	\item replacing $(X^{t_i,x}_{t_{i+1}},P^{t_i,p,x,\alpha^*}_{t_{i+1}})$ 
	by the Markov chain approximation $(X^{i,x}_{i+1}, P^{i,p,x,a^*_i}_{i+1})$, 
	obtained by the Euler scheme~\eqref{eq:XP};
	\item replacing the function $a^*_i$ by the piecewise constant approximation $\hat a_i$~\eqref{eq:approxa};
\end{enumerate}
For $i\in \{0,\cdots, N-1\}$, 
we define $\hat V^i$ the resulting approximation of $\hat V(t_i,\cdot,\cdot)$ 
satisfying the following  backward approximation scheme
\begin{equation}
\label{eq:approxw}
{\hat V}^i(x,p)= \E[{\hat V}^{i+1}(X^{i,x}_{i+1}, P^{i,p,x,\hat a_i}_{i+1})]\ .
\end{equation}
Let us assume, at this stage, that ${\hat V}^{i+1}$ is 
a given approximation of ${\hat V}(t_{i+1},\cdot,\cdot)$, 
which is two times continuously differentiable w.r.t. both variables. 
Now recall that $\hat a_i$ is supposed to be constant 
on $C_{i,r}$, for any $r\in \{1,\cdots, R\}$.  
Then, for any $(x,p)\in Int(C_{i,r})$,  $(X^{i,x}_{i+1},P^{i,x,p,\hat a_i}_{i+1})$ 
follows a log-normal distribution~\eqref{eq:XP} and 
we can apply tangent process approach \cite{broadie1996estimating} on \eqref{eq:approxw} 
to obtain that $\hat V^i$ is two times continuously differentiable and 
a backward formula for the derivatives
\begin{equation}
\label{eq:approxDerive}
 \left\{
 \begin{array}{ll}
 \hat V^i_p (x, p) &= \frac{1}{p}\E \left[ P^{i,x,p,\hat a_i}_{i+1} \hat V^{i+1}_p(X^{i,x}_{i+1}, P^{i,x,p,\hat a_i}_{i+1}) \right]\\
 \hat V^i_{xp} (x, p) &= \frac{1}{xp}\E \left[X^{i,x}_{i+1} P^{i,x,p,\hat a_i}_{i+1}  \hat V^{i+1}_{xp}(X^{i,x}_{i+1}, P^{i,x,p,\hat a_i}_{i+1}) \right]\\
 \hat V^i_p (x, p) &= \frac{1}{pp}\E \left[(P^{i,x,p,\hat a_i}_{i+1})^2  \hat V^{i+1}_{pp}(X^{i,x}_{i+1}, P^{i,x,p,\hat a_i}_{i+1}) \right] \; .\\
 \end{array}
\right. 
\end{equation}

\subsubsection{Piecewise constant regression and fixed point algorithm}

Besides, recall that $a^*_i$ is defined 
as a function of $\bar V_p(t_i,\cdot,\cdot)$, 
$\bar V_{px}(t_i,\cdot,\cdot)$ and  $\bar V_{pp}(t_i,\cdot,\cdot)$ 
according to equation~\eqref{eq:discrete optimal control}. 
Similarly, we want to impose the same relation between 
$\hat a_i$ and $\hat V^i_p$, $\hat V^i_{px}$ and  $\hat V^i_{pp}$. 
For this purpose, let us define the map $f\mapsto T_i(f)$ 
such that for any real valued function $f$ defined on $\R_+^*\x \R_-^*$, 
\begin{equation}
 \label{eq:fixed point operator}
 T_i ( f)(x,p):=\frac{ \theta \E[P^{i,x,p, f}_{i+1}\hat V^{i+1}_p (X^{i,x}_{i+1}, P^{i,x,p, f}_{i+1})] 
 - \sigma \E [X^{i,x}_{i+1} P^{i,x,p, f}_{i+1}\hat V^{i+1}_{xp} (X^{i,x}_{i+1}, P^{i,x,p, f}_{i+1}) ] }
 {\E [(P^{i,x,p, f}_{i+1})^2 \hat V^{i+1}_{pp} (X^{i,x}_{i+1}, P^{i,x, p,  f}_{i+1})]}\ ,
\end{equation}
for all $(x,p)\in \R_+^*\x \R_-^*$. 
Notice that the map $T_i$, defined above, 
depends implicitly on the previous approximations
$\hat V^{i+1}_p$, $\hat V^{i+1}_{px}$ and $\hat V^{i+1}_{pp}$. 
Then $\hat a_i$ could be obtained as a piecewise constant 
approximation of a fixed point of $T_i$.  

One way to do this, 
is to approximate $T_i$ by $\hat T_i$, obtained in replacing  
the conditional expectation, in~\eqref{eq:fixed point operator}, 
by a regression operator, $\hat \E_i$, on a  set of regression functions 
which are  piecewise constant. Consider for instance the following set 
of regression functions $(\mathbf{1}_{C_{i,r}})_{r=1,\cdots , R}$. 
We introduce
\begin{equation}
 \label{eq:hatTi}
 \hat T_i ( f)(x,p):=\frac{ \theta \hat \E_i[P^{i,x,p, f}_{i+1}\hat V^{i+1}_p (X^{i,x}_{i+1}, P^{i,x,p, f}_{i+1})] 
 - \sigma \hat\E_i [X^{i,x}_{i+1} P^{i,x,p, f}_{i+1}\hat V^{i+1}_{xp} (X^{i,x}_{i+1}, P^{i,x,p, f}_{i+1}) ] }
 {\hat \E_i [(P^{i,x,p, f}_{i+1})^2 \hat V^{i+1}_{pp} (X^{i,x}_{i+1}, P^{i,x, p,  f}_{i+1})]}
 \end{equation}
so that  $\hat T_i ( f)$ is automatically piecewise constant on the partition $({C_{i,r}})_{r=1,\cdots , R}$. 

Adding up all these  approximations, we finally obtain, at each point, 
$t_i$, of the mesh grid, $\pi$,  an approximation
$(\widehat V^i,\widehat V^i_p,\widehat V^i_{px},\widehat V^i_{pp}, \hat a_i)$ of 
$$
\big (\bar V(t_i,\cdot,\cdot),\bar V_p(t_i,\cdot,\cdot), \bar V_{px}(t_i,\cdot,\cdot), \bar V_{pp}(t_i,\cdot,\cdot), a^*(t_i,\cdot,\cdot)\big)
$$ 
by the following algorithm applied with a fixed  tolerance parameter $\eps>0$:
\begin{description}
	\item[Initialization]
\begin{equation}
 \label{eq:algoInit}
 \left\{
 \begin{array}{ll}
 \widehat V^N (x, p) &= \Xi^{-1} (x, p)\\
 \widehat V^N_p (x, p) &= \Xi^{-1}_p (x, p)\\
 \widehat V^N_{px} (x, p) &= \Xi^{-1}_{px}(x, p)\\
 \widehat V^N_{pp} (x, p) &= \Xi^{-1}_{pp}(x, p)\\
 \hat a_N(x,p)&=\frac{\theta p \widehat V^N_p(x,p) - \sigma x p \widehat V^N_{xp}(x,p)}{pp \widehat V^N_{pp}(x,p)} \ .
 \end{array} \right. 
\end{equation}
\item[From step A($N-1$) to A($0$) :]
\end{description}
\begin{itemize}
	\item [A($i$) : ]\hspace{0.2cm} \texttt{ SET } $a:=\hat a_{i+1}$;	\texttt{ GOTO} B($i,a$);
	\item [B($i,a$):]\hspace{0.2cm} 
	\begin{enumerate}
	\item \texttt{SET }
	$
	a':=\hat{T}_{i}(a)\quad (\textrm{recall that}\  \hat T_i \ 
	\textrm{depends on $\widehat V^{i+1}_p$, $\widehat V^{i+1}_{px}$ and $\widehat V^{i+1}_{pp}$})
	$;
	\item \texttt{IF} $\vert a'-a\vert \leq \eps$ 
	\begin{itemize}
	\item \texttt{THEN SET}  
$$
 \left\{
 \begin{array}{ll}
  \hat a_i&=a'\\
\widehat V^i(x,p)&= \hat \E_i[\widehat V^{i+1}(X^{i,x}_{i+1}, P^{i,p,x,\hat a_i}_{i+1})]\\
\widehat V^i_p (x, p) &= \frac{1}{p}\hat \E_i \left[ P^{i,x,p,\hat a_i}_{i+1} \widehat V^{i+1}_p(X^{i,x}_{i+1}, P^{i,x,p,\hat a_i}_{i+1}) \right]\\
\widehat V^i_{xp} (x, p) &= \frac{1}{xp}\hat \E_i \left[X^{i,x}_{i+1} P^{i,x,p,\hat a_i}_{i+1}  \widehat V^{i+1}_{xp}(X^{i,x}_{i+1}, P^{i,x,p,\hat a_i}_{i+1}) \right]\\
\widehat V^i_p (x, p) &= \frac{1}{pp}\hat \E_i \left[(P^{i,x,p,\hat a_i}_{i+1})^2  \widehat V^{i+1}_{pp}(X^{i,x}_{i+1}, P^{i,x,p,\hat a_i}_{i+1}) \right] \; .
\end{array} \right. 
$$
	\begin{itemize}
		\item \texttt{IF} $i=0$ \texttt{THEN STOP}; 
		\item \texttt{ELSE GOTO} A($i-1$); 
	\end{itemize}
	\item	\texttt{ELSE GOTO} B($i, a'$);
	\end{itemize}
	\end{enumerate}
\end{itemize}
Notice that limiting the previous algorithm to one fixed point iteration (in $B(i,a)$) 
reduces to an explicit scheme for which  $\hat a_i$ is given as a function of 
the derivatives at the next time step $t_{i+1}$, 
$(\widehat V^{i+1},\widehat V^{i+1}_p,\widehat V^{i+1}_{px},\widehat V^{i+1}_{pp})$ 
and the  control $\hat a_{i+1}$. 
However, 
in practice at most three iterations are sufficient 
to obtain reasonable convergence to the fixed point. 
In theory, the contraction of $\hat T_i$ should be proved for 
a sufficiently small time step $\Delta t$. 
But this is left for future works.
If we don't proceed to a convergence analysis of
the scheme with $N$, for a general diffusion
and general loss function, we can still provide
a numerical confirmation of the relevance of the method. 
To validate the algorithm we proceed in Section \ref{sec:validation}
to a comparison between the explicit formula of Corollary \ref{prop:solution in complete market}
and the value provided by the algorithm.

\section{Numerical tests}
\label{sec:validation}

The present section is devoted to tests 
on real and simulated data meant to illustrate 
the interest of the partial hedging strategy developped in 
this article and to validate the numerical scheme 
introduced in the previous section. 
We proceed into four steps. 
First, we  fit the parameters of the exponential model \eqref{eq:price dynamics} on real data.
Then, we point out the importance of the risk induced by the random shaping factor, 
by  evaluating the hedging error implied, on real data, 
by the naive Black-Scholes hedging strategy based on a prediction of the shaping factor 
(without taking into account its randomness). 
This naive hedging approach will constitute our benchmark. 
To validate  the numerical approximation scheme introduced in 
Section \ref{sec:numerics}, 
we analyse its performance on the explicit case of Section~\ref{sec:complete}.
We finally compare the partial hedging procedure, on simulations, 
to the benchmark that shall be introduced right away.

\subsection{Black-Scholes benchmark}
\label{sec:BS}

We consider the following Black-Scholes strategy for a naive agent.
The naive agent assumes that the set $L$ reduces
to a singleton $\{\lambda_0\}$.
This belief is accepted for example as a raw
approximation of the expected value of $\Lambda$.
In this situation,
the previous setting reduces to the case of Section \ref{sec:complete}.
However, since the market is complete,
the naive agent desires to put in place a
complete hedging strategy allowed by the Black-Scholes
framework.
The naive benchmark is thus given by
\begin{enumerate}
 \item an initial value provided by the Black-Scholes price
 of the contingent claim $g(\lambda_0 X_{T^*}^{t,x})$, given by
 equation \eqref{eq: black-scholes price}: $C(t,x,\lambda_0)$.
 
 \item A hedging strategy, associated to that belief,
 and given by the delta-hedging procedure
 $\nu_s = C_x (s, X^{t,x}_s, \lambda_0)$ on $[0,T]$.
 After the apparition at $T$ of the asset $\Lambda X$,
 the option price is impacted immediately by
 the real value taken by $\Lambda$, different from $\lambda_0$,
 but the portfolio stays self-financed, and continuous.
 We assume here that the naive agent continues with
 the delta-hedging strategy until $T^*$.
 
 \item A terminal hedging error that spreads from time $T$
 with value $\eps_T: = C(T, x, \lambda_0) - C(T, x, \Lambda)$.
 In the Black-Scholes setting, by a simple no-arbitrage argument and wero interest rate,
 this error remains constant until $T^*$.
\end{enumerate}

The motivation of such a strategy is to average the losses
by averaging the possible values taken by $\Lambda$.
This is however wrong as the price of the derivative
is mostly a non-linear function of the underlying price.
In the studied example of the Call option below,
if $\lambda$ is fixed, we obtain
a non-linear function of the strike:
$$
\begin{array}{ll}
C_{BS}(t,\lambda_0 x,K) &
:= \E^{\Q}[(\lambda_0 X^{t,x}_{T^*}-K)^+ ]
= \lambda_0 \E^{\Q}[(X^{t,x}_{T^*}-K /\lambda_0)^+ ] \\
& = \lambda_0 C_{BS}(t,x, K/ \lambda_0).
\end{array}
$$

\subsection{Analysis on real data}
\label{sec:error}

To provide a realistic framework,
we refer to historical data.
This allows to propose a model for $L$ and $\Lambda$,
and values for parameters $(\mu, \sigma)$ of the exponential dynamics~\eqref{eq:price dynamics}.
The available data designates
daily quotations of
futures prices on the French Power Market,
provided by EEX. 
We consider a delivery period covering the period 
 from October 2004 
to March 2011, i.e., 
78 Month delivery Futures 
during their whole quotation period 
and the respective Quarter delivery futures
contracts covering them.
Two estimations are made out of it.
\begin{enumerate}
 \item This provides 78 observations for
a supposedly repeated realisation of the
random variable $\Lambda$.
The average is $\hat \Lambda=1.0012$
and its variance $V(\Lambda)=0.081$.
We then assume that $\Lambda$ follows a 
scaled beta law with these characteristics:
$\Lambda \sim 3 \beta (114,227)$. 
This is justified by the fact that $\Lambda$
shall have a bounded support, which is assumed here to be
the interval $[0,3]$.

\item The parameters $\mu$ and $\sigma$ in
the exponential dynamics \eqref{eq:price dynamics} are computed
on the aggregated returns of month futures and quarter futures.
Here, $\mu$ contains the discount rate (since
we assumed that the interest rate is null by omitting it).
The obtained drift $\hat \mu$ is null, and the obtained (yearly)
volatility $\hat \sigma=28\%$.
\end{enumerate}

To quantify the impact of neglecting the uncertainty on the shaping factor $\Lambda=\lambda^*$, on the performance of the hedging strategy, we have implemented, on real data, the naive  Black-Scholes hedging strategy supposing different given parameters $\lambda$ varying around
the real observed value $\lambda^*$ with an amplitude of error of $50\%$. In our tests, we have considered call options on our $78$ Month delivery futures with various maturities ans strikes as indicated on Figure \ref{fig:errors cumulated}. 
The resulting hedging error can be decomposed into four sources: 
\begin{enumerate}
	\item  hedging at discrete times (the Delta hedging strategy is indeed implemented daily);
	\item  errors on the dynamical model or on the parameters  (the hedging instrument may not have  i.i.d. log-returns with log-normal distributions, $\hat \mu$ and $\hat \sigma$ are only estimations);
	\item the limited number of hedging scenario inducing a statistical error;
	\item error on the shaping factor value.
\end{enumerate}
The scope of this paper focuses specifically on the latter source of error. 
Hence, to distinguish the contribution of each error and to separate the fourth one, 
we have represented in Figure~\ref{fig:errors cumulated}  four quantities: 
\begin{enumerate}
  \item \textit{Real error} We evaluate the hedging error  using
 the naive strategy on real data. 
  \item \textit{Simulated Error (78 data)} We then do the same on simulated data on the same time grid,
 and with the same number of trajectories (78).
 This allows to quantify the error due to 
 the model and the parameters estimation,
 by comparing it to the previous error.
\item \textit{Simulated Error (78x100 data)} We repeat this procedure with a greater number
 of simulations $(78\times 100)$ in order to confirm that
 the previous error is not erroneous because of the low number
 of studied trajectories.
 \item \textit{Theoretical error} We represent the error induced by hedging in the Black-Scholes framework (in continuous time) with the wrong shaping factor. 
 This represents exclusively the hedging error due to the error on the shaping factor. Notice that for $t\geq T$, $\lambda^*$ is known, hence on $[T,T^*]$ the naive Black-Scholes strategy is equal to the complete market replication strategy that would have been implented from time $0$ if $\lambda^*$ was known. Hence the theoritical hedging error reduces to the difference between the values at time $T$ of the naive hedging portfolio (with the wrong value of $\lambda$) and the perfect hedging portfolio (with the right value of $\lambda=\lambda^*$): $C_{BS}(t,\lambda x,K)-C_{BS}(t,\lambda^* x,K)$. 
\end{enumerate}
Altogether,
the results presented in Figure \ref{fig:errors cumulated}
push to the following temporary conclusions.
An error on the value of $\Lambda$ can
significantly impact the error.
After that, the main error is due to 
the discretization of the hedging strategy.
Hence, it is worth developping a specific methodology to take into account the uncertainty of the shaping factor in the hedging strategy. 

\subsection{Convergence of the approximation scheme to explicit solution}
\label{sec:convergence}

We shall test the efficiency of
the algorithm presented in Section \ref{sec:numerics}.
To do so, we compute the hedging strategy
and the value function in the specific case of
Section \ref{sec:interval complete}.
Assume without loss of generality that $\Lambda=1$.
Recall that we obtain the following 
explicit expression for the $\P$-martingale 
$P^{t,p,\alpha}$ initialized at time 
$t\ge T$ and the function $v$ for any $s\in [t,T^*]$,
\begin{equation}
\label{eq:P}
\left \{
\begin{array}{l}
v(s,X^{t,x}_s,P^{t,p,\alpha}_s, 1)=C(s,X^{t,x}_s, 1)-(-k P^{t,p,\alpha}_s)^{1/k} \exp\brace{\frac{\theta^2}{2(k-1)}(T^*-s)}\\
P^{t,p,\alpha}_s=p\left (\frac{X^{t,x}_s}{x}\right )^{-\frac{k}{k-1}\frac{\mu}{\sigma^2}}
\exp\brace{\frac{k^2 ( \theta^2 -\mu)}{2(k-1)} (s-t)}\ , .
\end{array}
\right .
\end{equation}
Observe that $P^{t,p,\alpha}_{s}$ can be expressed as a function of $X^{t,x}_s$ i.e. $P^{t,p,\alpha}_s =p(t,s,X^{t,x}_s)$.
Hence we analyse the performance of our algorithm by observing its ability to approximate 
the one dimensional real valued function $u_s$ such that 
\begin{equation}
\label{us}
u_s(x)=V(s,x,p(t,s,x))\ .
\end{equation} 
In our simulations, we consider the following parameters.
\begin{enumerate}
 \item The model parameters are slightly modified
 with values $(\hat \mu, \hat \sigma)=(0.1, 0.28)$.
 The initial asset price value is fixed at $x=50.89$.
 \item The convexity parameter of the loss function
 is $k=2$ and the level $p$ takes the value $0.1$ Euro$^2$.
 \item The option is a call option with a strike $K$
 and a maturity of $20$ trading days, i.e., $T=20/250$.
 \item We have performed our algorithm with 
 $M=10^5$ particles to estimate at each step of time 
 the conditional expectations and 
 a time discretization mesh $t_0=0,\cdots t_i,\cdots t_N=T$ with a time step $\Delta t =1/250$.
\end{enumerate}
In our tests, the fixed point algorithm was limited to three iterations.
We have represented
on Figure~\ref{fig:robustness}
the value of $u_s(x)$ with respect to $x$
computed by the explicit formula and the numerical algorithm.
We also provide the value of the control $\nu$
at the initial date to illustrate the convergence of derivatives too.


\subsection{Performance with a call option}
\label{sec:performance}

We now compare the loss approach
(hereafter denoted \emph{shortfall risk}, or SR)
and the benchmark strategy
(hereafter \emph{Black-Scholes}, or BS)
upon a call option.
For each approach, 
we implement the associated hedging strategy 
on i.i.d. $M_{hedge}=10000$ simulated price paths. 
For each path we compute both hedging errors.
Then we compute
by Monte Carlo approximation 
(on these i.i.d. $M_{hedge}=10000$ simulations) 
the expected loss associated to 
the Black-Scholes approach and the shortfall risk hedge. 
Recalling Section \ref{sec:error},
the trading strategies are not implemented continuously and
the resulting hedging errors may differ
from the theoretical time continuous setting.
\begin{enumerate}
    	\item The naive Black-Scholes strategy
    	is settled with the value $\lambda_0 = \Esp{\Lambda}=1.0012$.
    	The variable $\Lambda$ is given by 
	a law $\beta(114,227)$, and the price model
	as in Section \ref{sec:convergence}.

	\item For the option, 
	we compare several strike possibilities for the Call option: $K = \gamma \lambda_0 x$
	with $\gamma$ taking values in the set $\brace{0.85 ; 0.9 ; 0.95 ; 1 ; 1.05, 1.1 ; 1.15 ; 1.2}$.
	
	\item The loss function is 
	the partial moment loss function of 
	Section \ref{sec:numerics}
	with $k=2$, and the threshold $p$ 
	varies enough to evaluate its impact.
	In the following comparison,
	we consider the square root 
	of the obtained error in order
	to express it in euros.
	This justifies the terminology \emph{shortfall},
	which is a monetary homogeneous quantity.

	\item Strategies are rebalanced daily, $T=128$ days and $T^*=184$ days, one year corresponds to $250$ days and $X_0=50.89$. 
\end{enumerate}

Figure \ref{fig:SR_vs_BS} sums up the simulations
and compares, for the different values of $K$,
the value function as a function of $p$.
Figure \ref{fig:SR_vs_BS_Cvar} provides a comparison
between the two approaches for another criterion:
the conditional Value-at-Risk, or expected shortfall.
These two figures lead us to the two following conclusions.
The first one is that
the partial hedging procedure SR
allows to hedge the quadratic loss function
more efficiently and with less initial amount of 
money than the BS strategy.
The second figure illustrates the fact that
this new strategy stays more interesting for another risk criteria than the one used in the specified control problem, 
providing also possible robustness of
the consideration of the shaping factor $\Lambda$
in our model.

\begin{figure}[p]
\centering
	\subfigure[$K=0,8 X_T$]
	{\includegraphics[width=11cm, height=5cm]
	{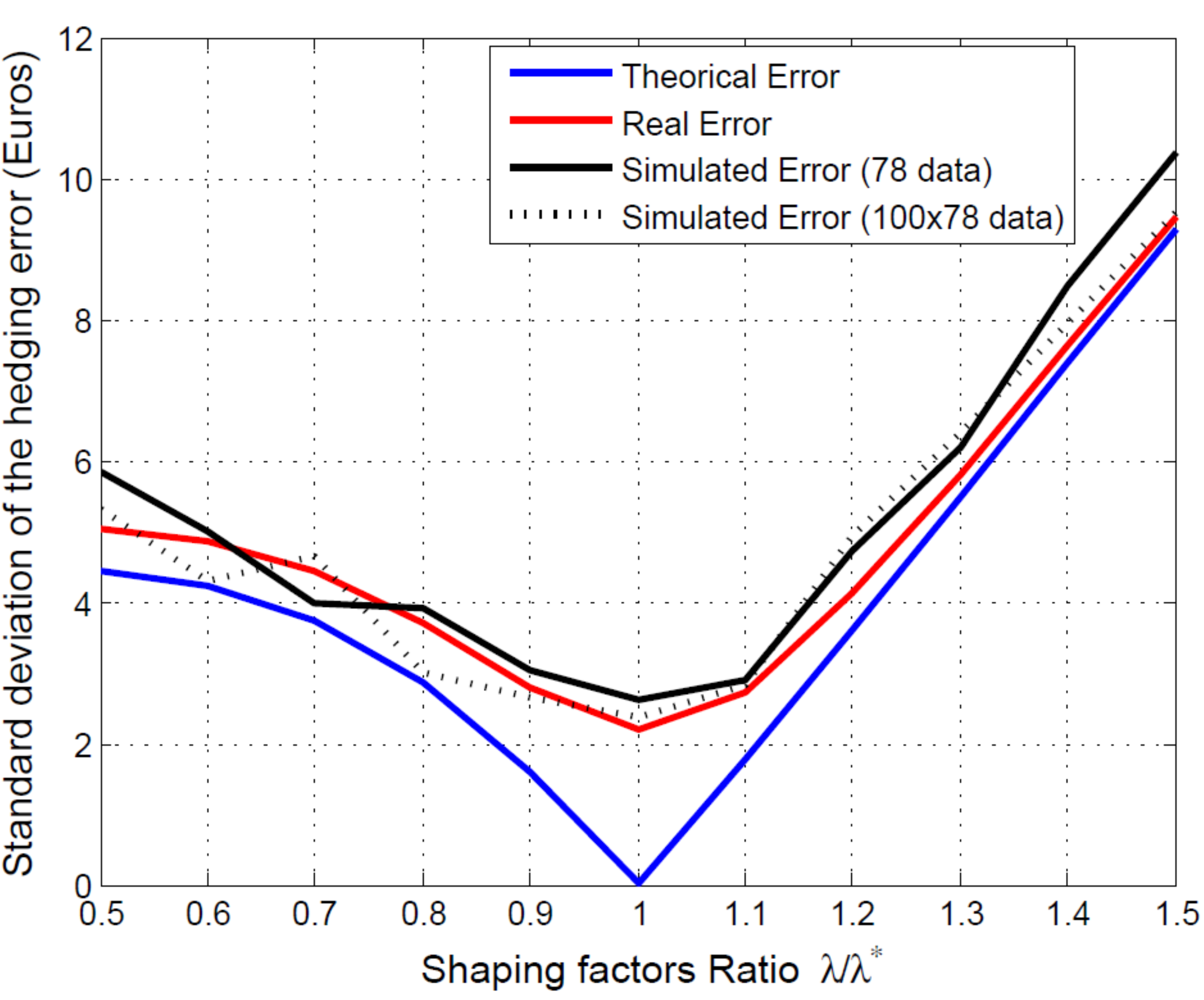}}
	\subfigure[$K=X_T$]
	{\includegraphics[width=11cm, height=5cm]
	{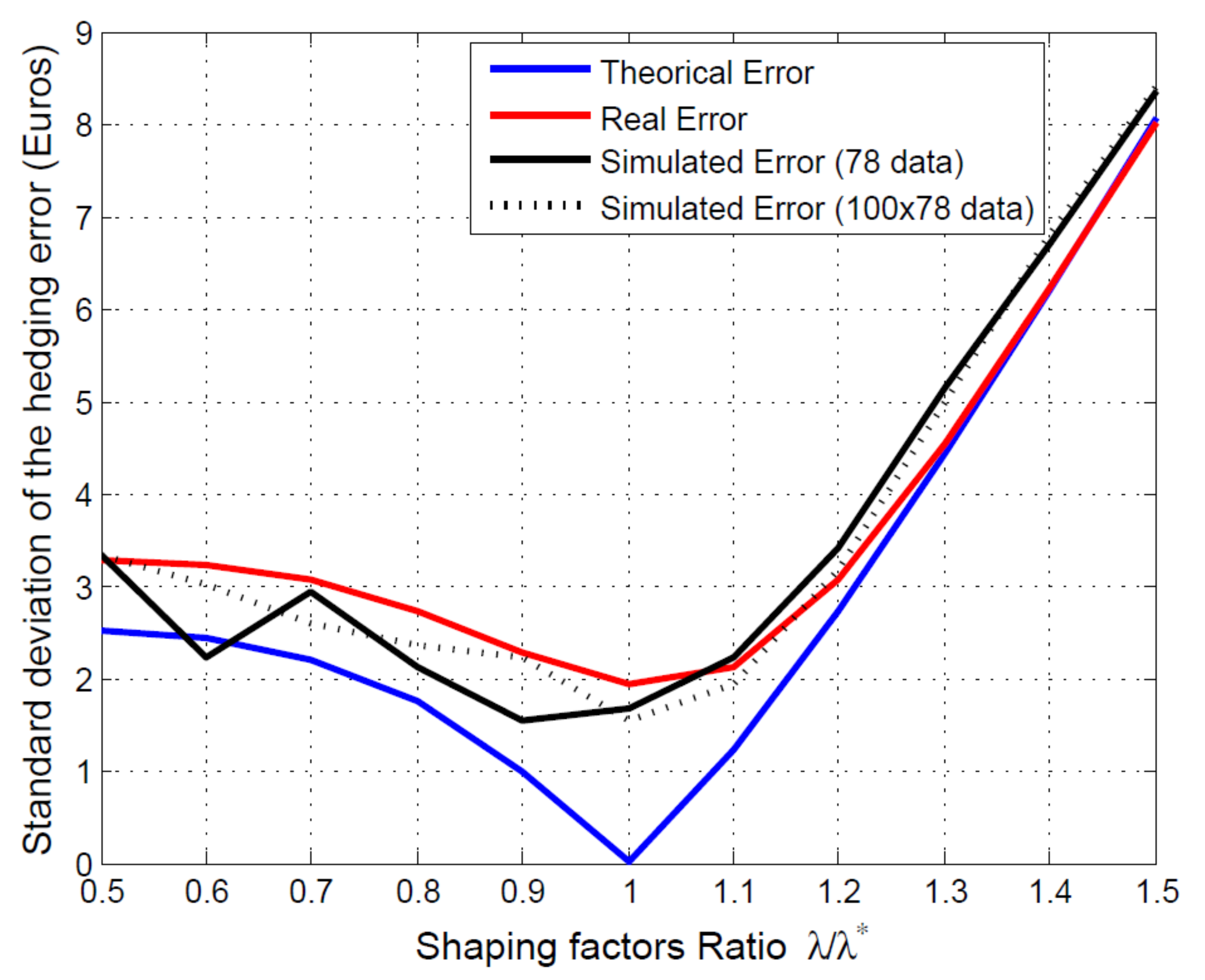}}
	\subfigure[$K=1.2 X_T$]
	{\includegraphics[width=11cm, height=5cm]
	{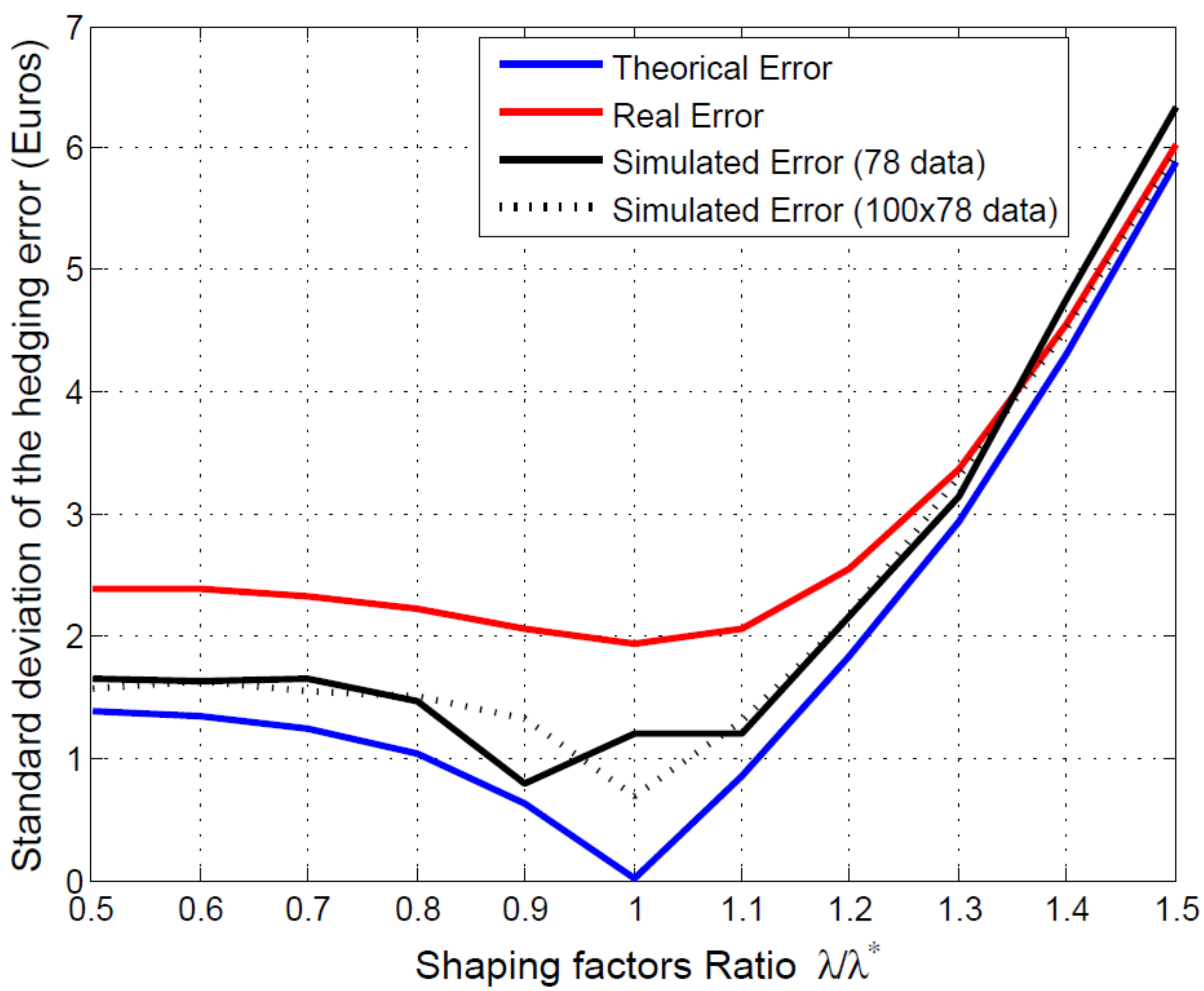}}

	\caption{\small Standard deviation of the hedging error as a function of the ratio between the shaping factor used in the hedging strategy $\lambda$ and the real shaping factor $\lambda^*$ impacting the historical scenarios.}
	\label{fig:errors cumulated}
\end{figure}

\begin{figure}[p]\centering
	\subfigure[Value function $x\mapsto u_s(x)$ with $p=0.1$ Euros and time steps $i=\in\{1,5,10,15,19\}$.]
	{\includegraphics[width=11cm, height=7.5cm]
	{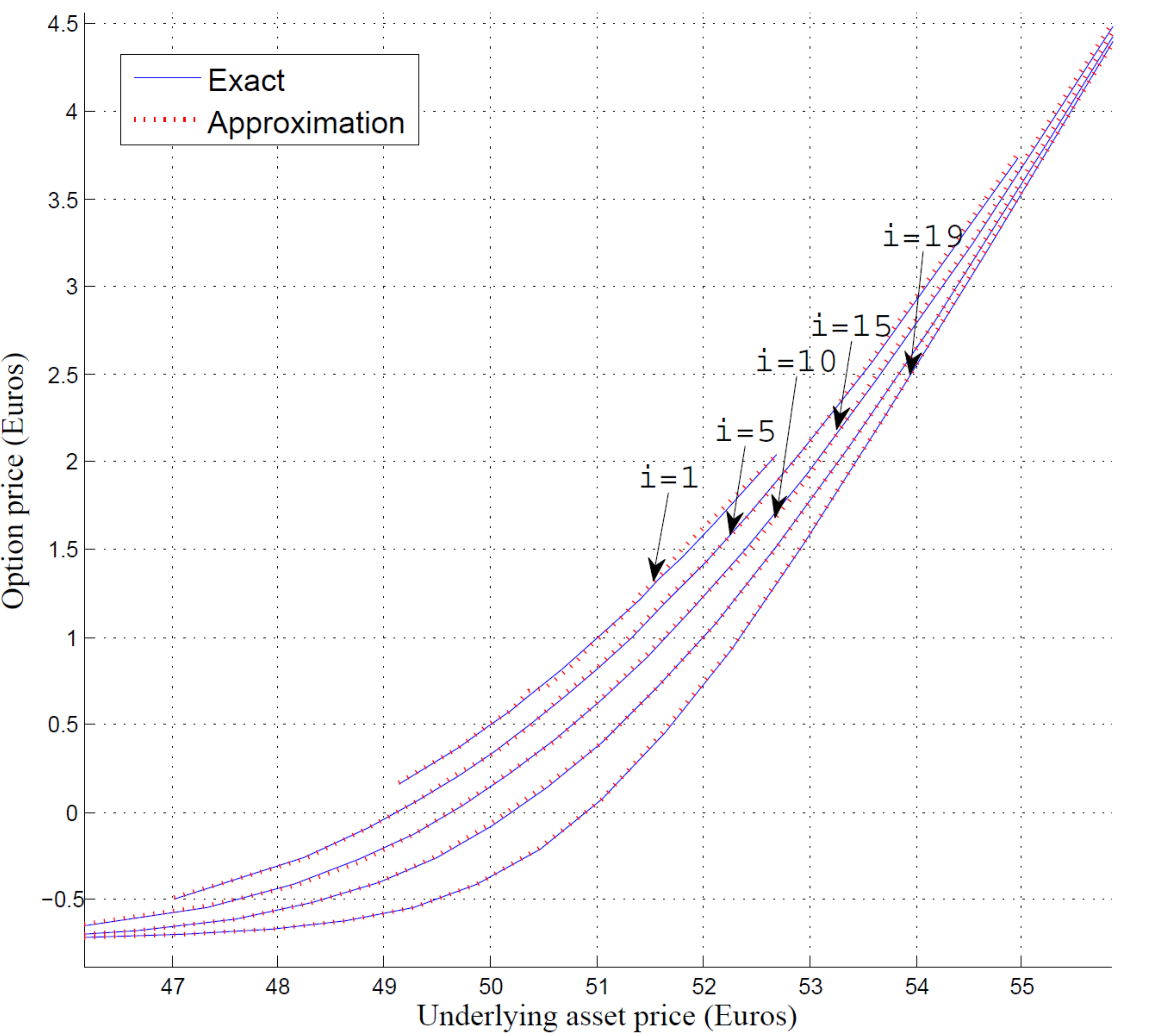}}
	
	\subfigure[Optimal strategy $x\mapsto \nu(t,x,p)$ for $p=0.1$ Euros and time steps $i\in\{1,5,10,15,19\}$.]
	{\includegraphics[width=11cm, height=7.5cm]
	{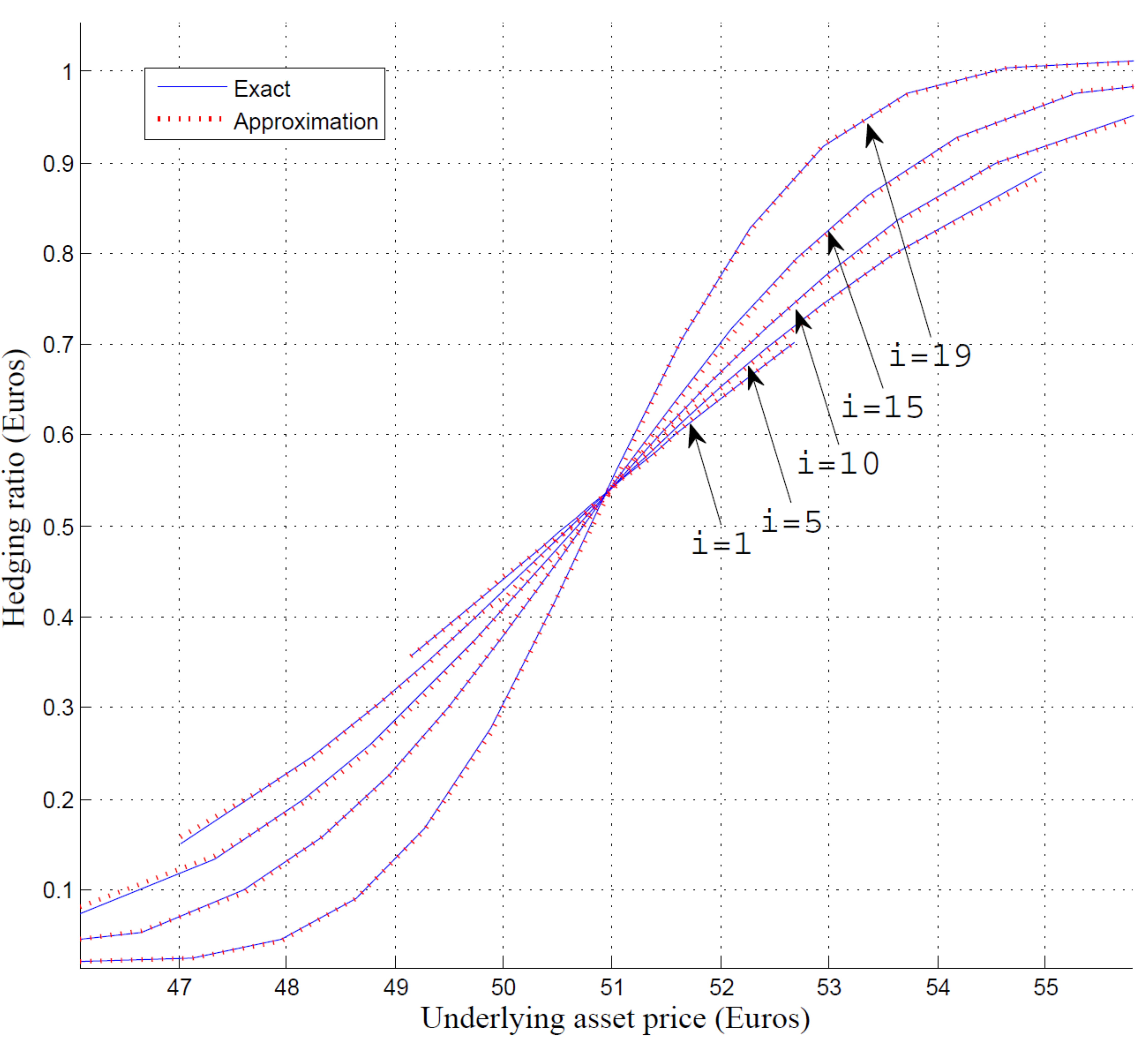}}

	\caption{\small Comparison between the numerical solution and
	the explicit formula.}
	\label{fig:robustness}
\end{figure}

\begin{figure}[p]
\begin{center}
\subfigure[$K=0,85 X(t_0)$]
{\includegraphics[width=5.5cm, height=3.2cm]
{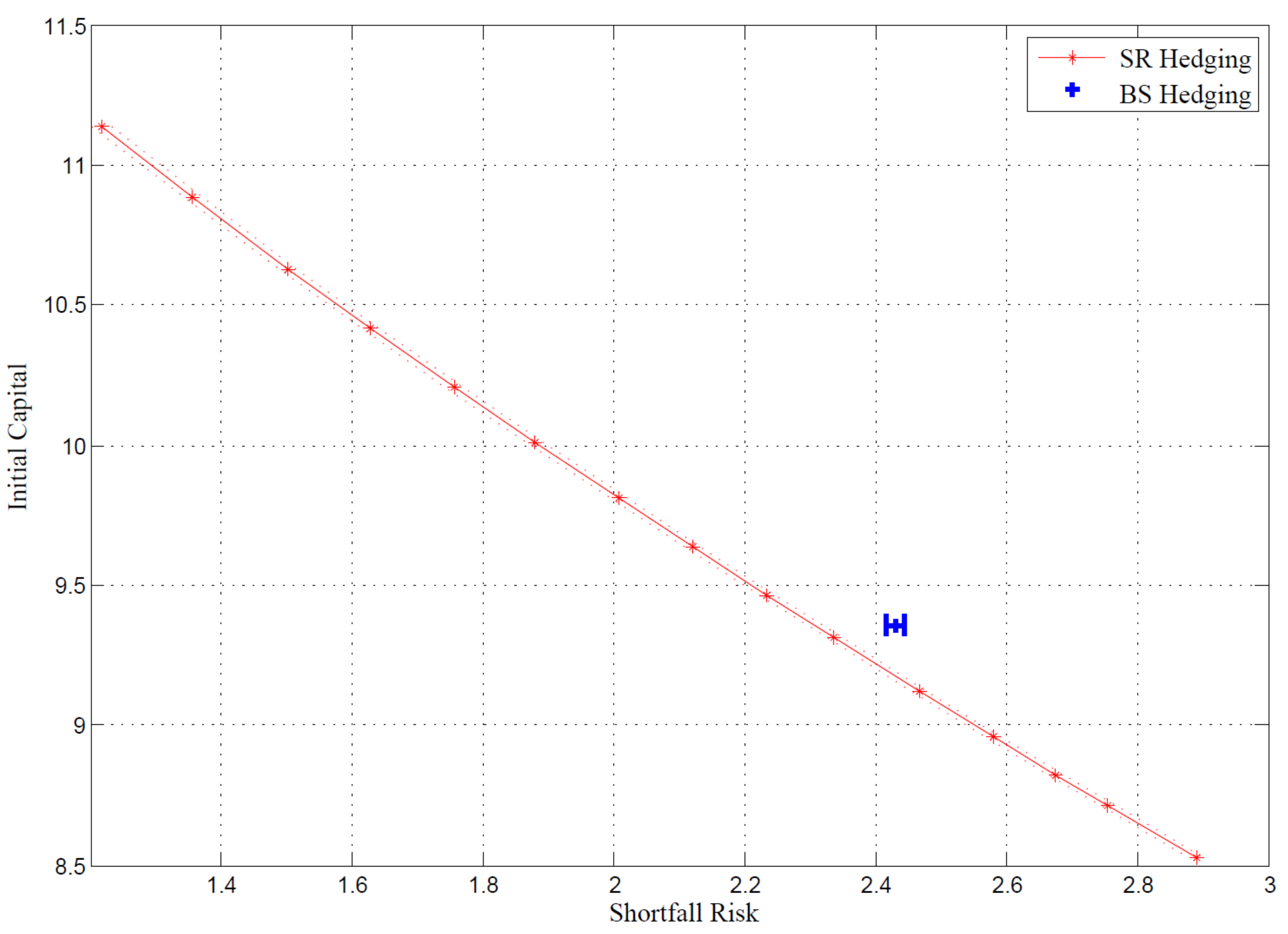}}
\subfigure[$K=0,90 X(t_0)$]
{\includegraphics[width=5.5cm, height=3.2cm]
{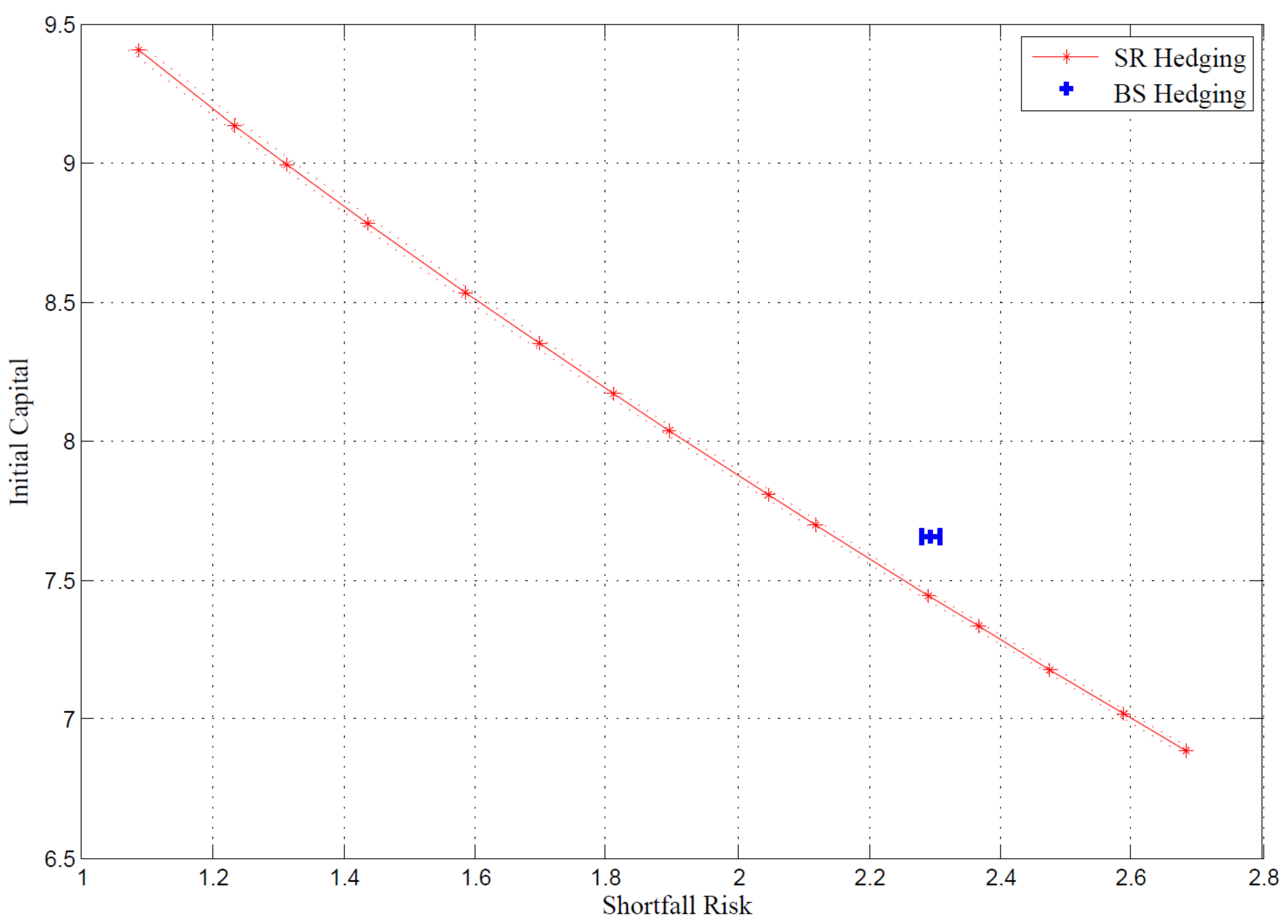}}
\subfigure[$K=0,95 X(t_0)$]
{\includegraphics[width=5.5cm, height=3.2cm]
{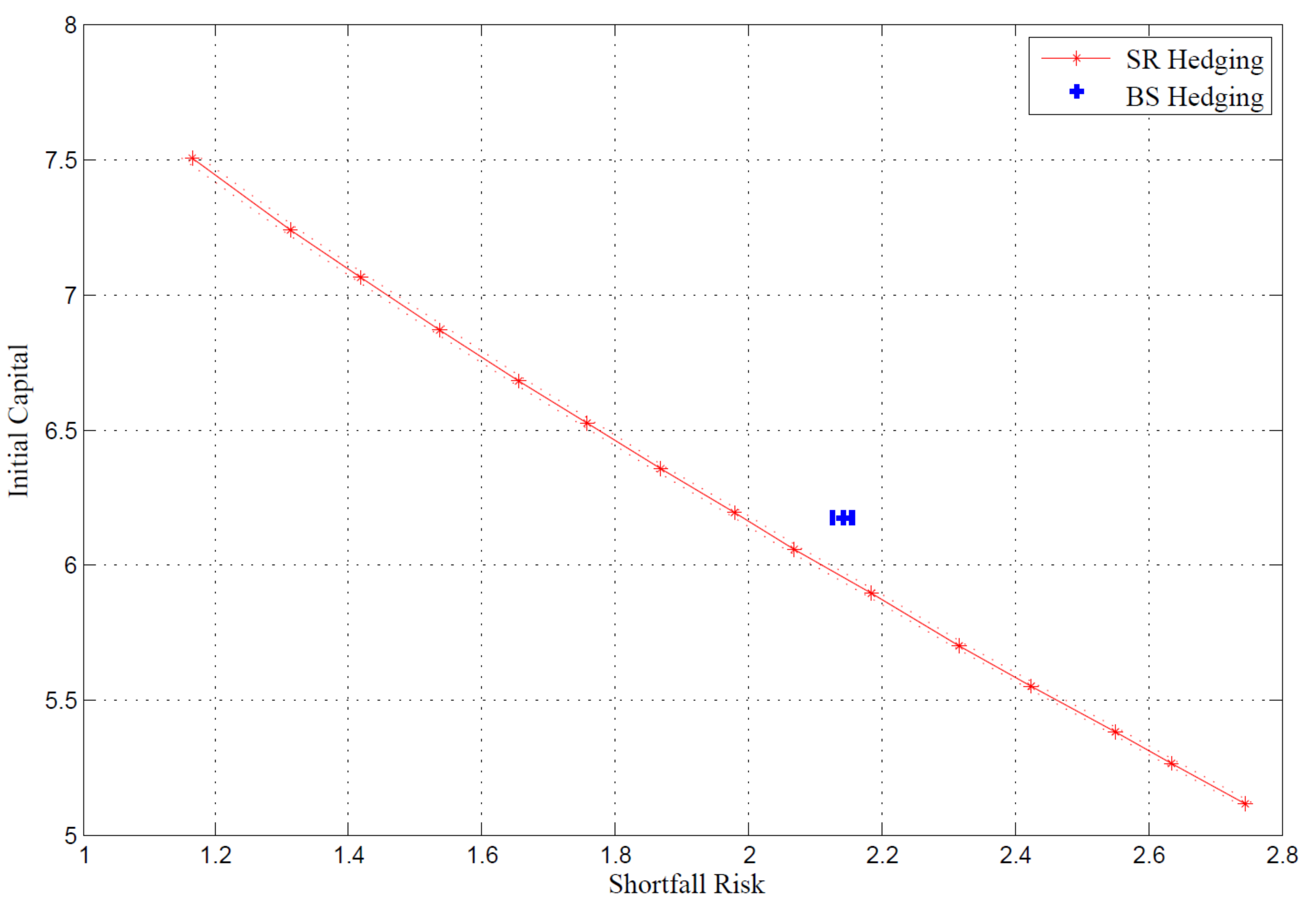}}
\subfigure[$K= X(t_0)$]
{\includegraphics[width=5.5cm, height=3.2cm]
{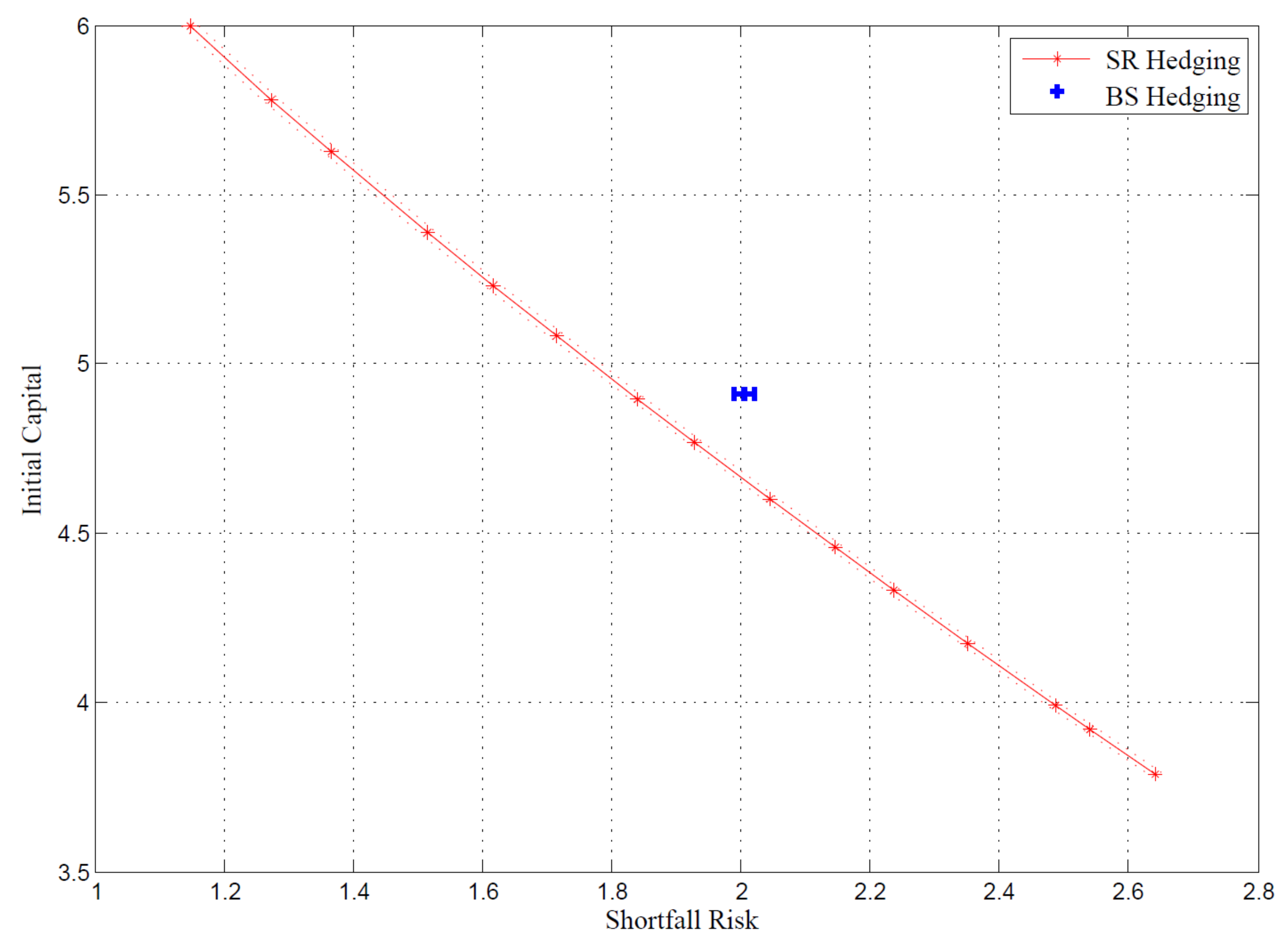}}
\subfigure[$K=1,05 X(t_0)$]
{\includegraphics[width=5.5cm, height=3.2cm]
{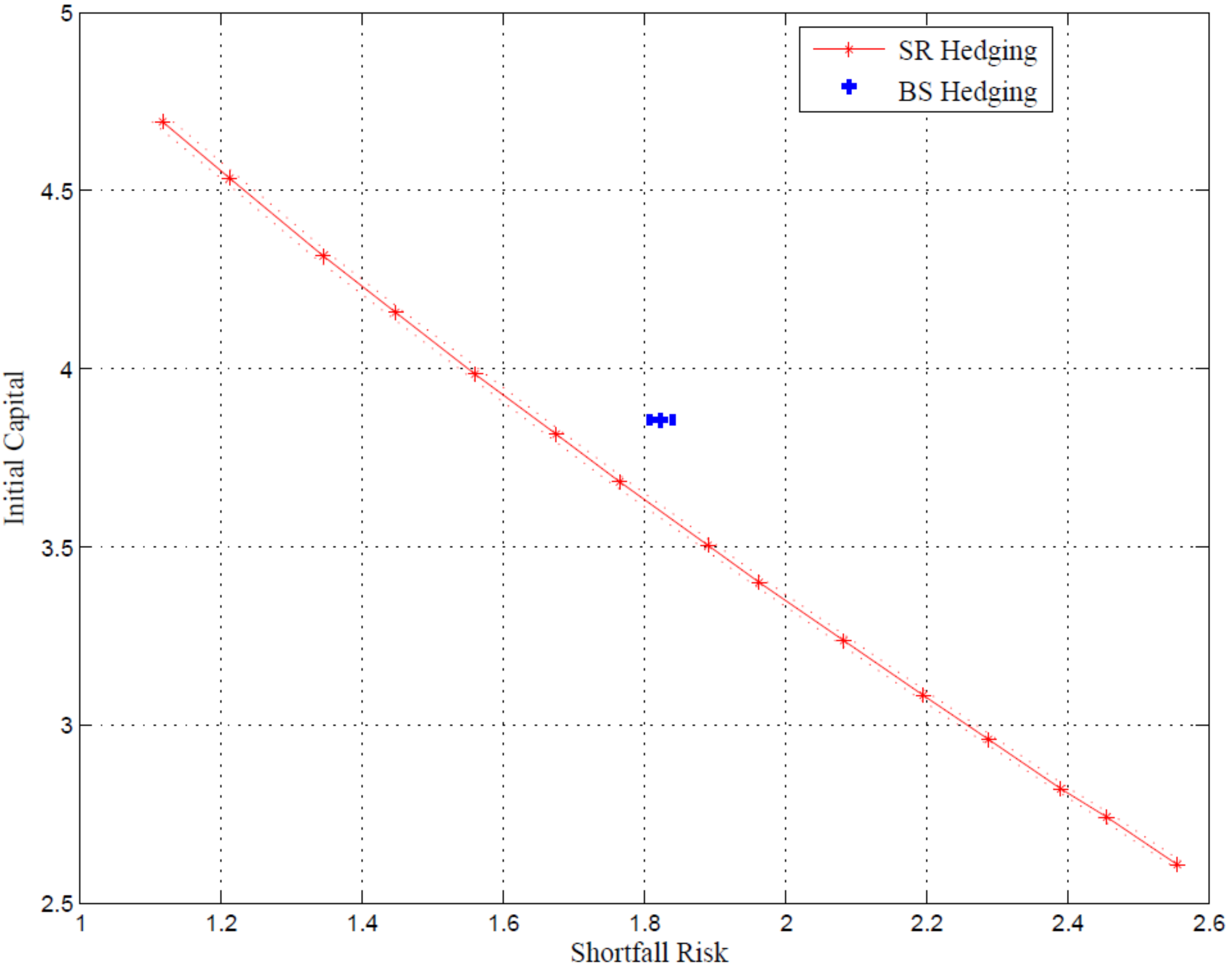}}
\subfigure[$K=1,10 X(t_0)$]
{\includegraphics[width=5.5cm, height=3.2cm]
{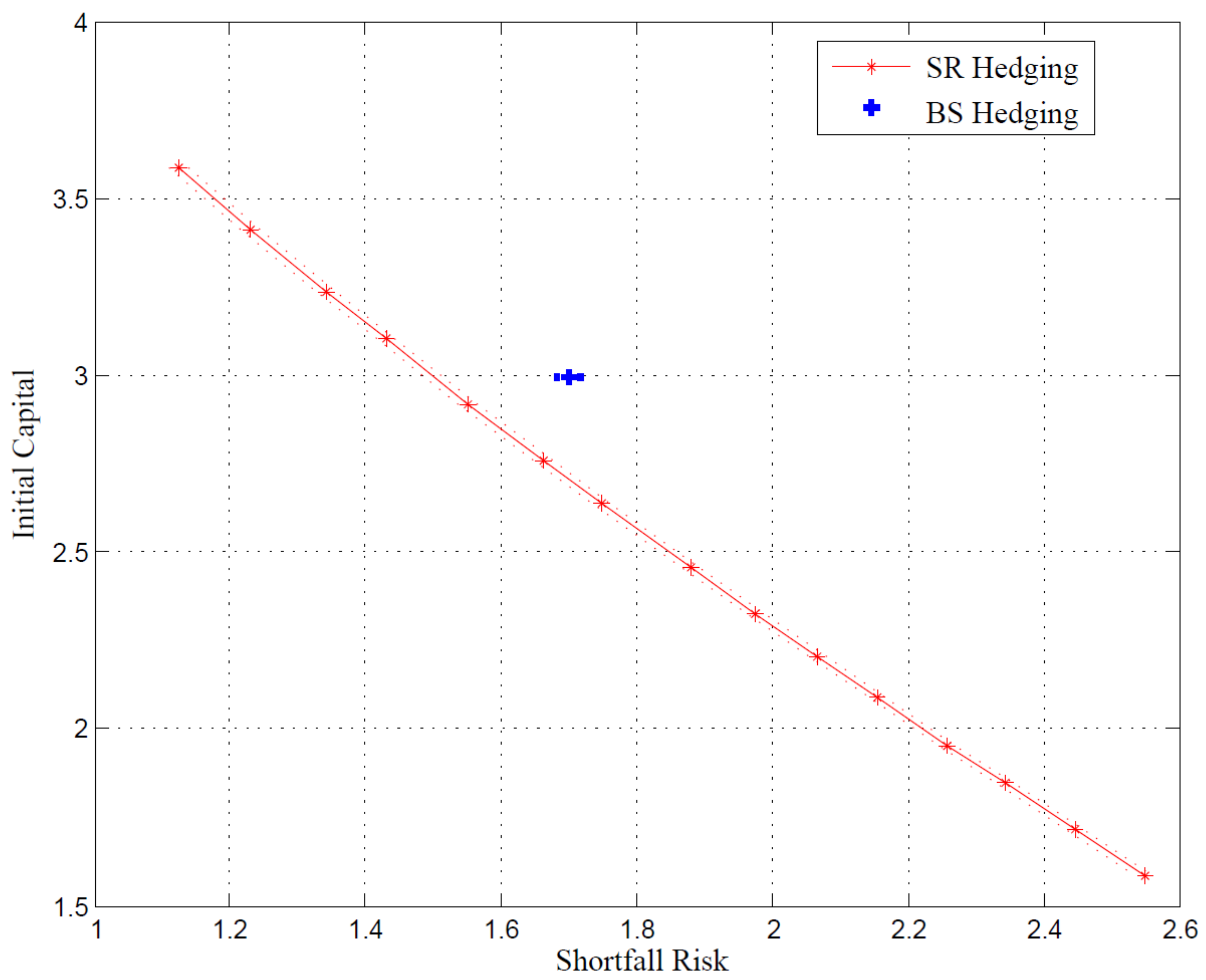}}
\subfigure[$K=1,15 X(t_0)$]
{\includegraphics[width=5.5cm, height=3.2cm]
{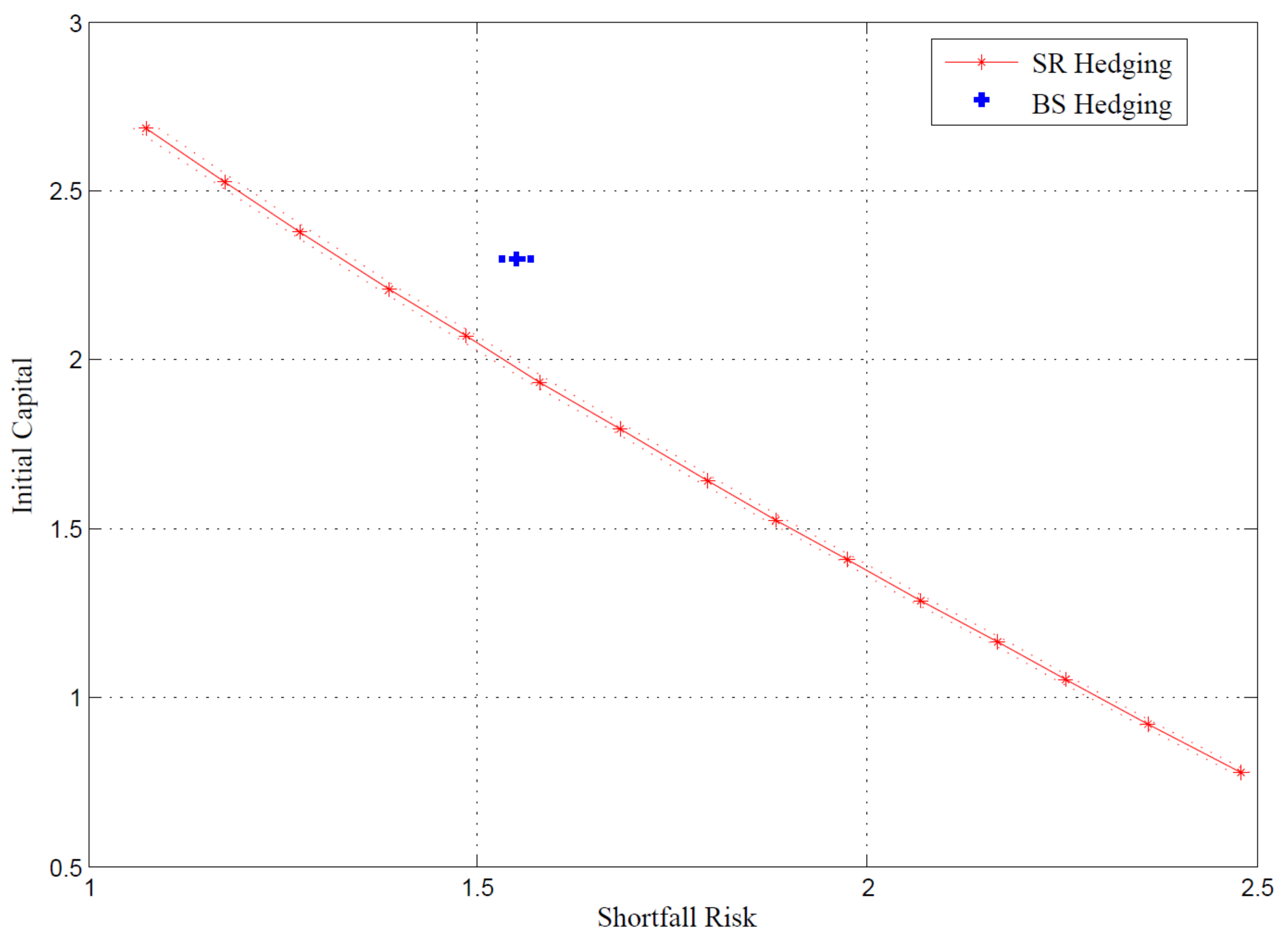}}
\subfigure[$K=1,15 X(t_0)$]
{\includegraphics[width=5.5cm, height=3.2cm]
{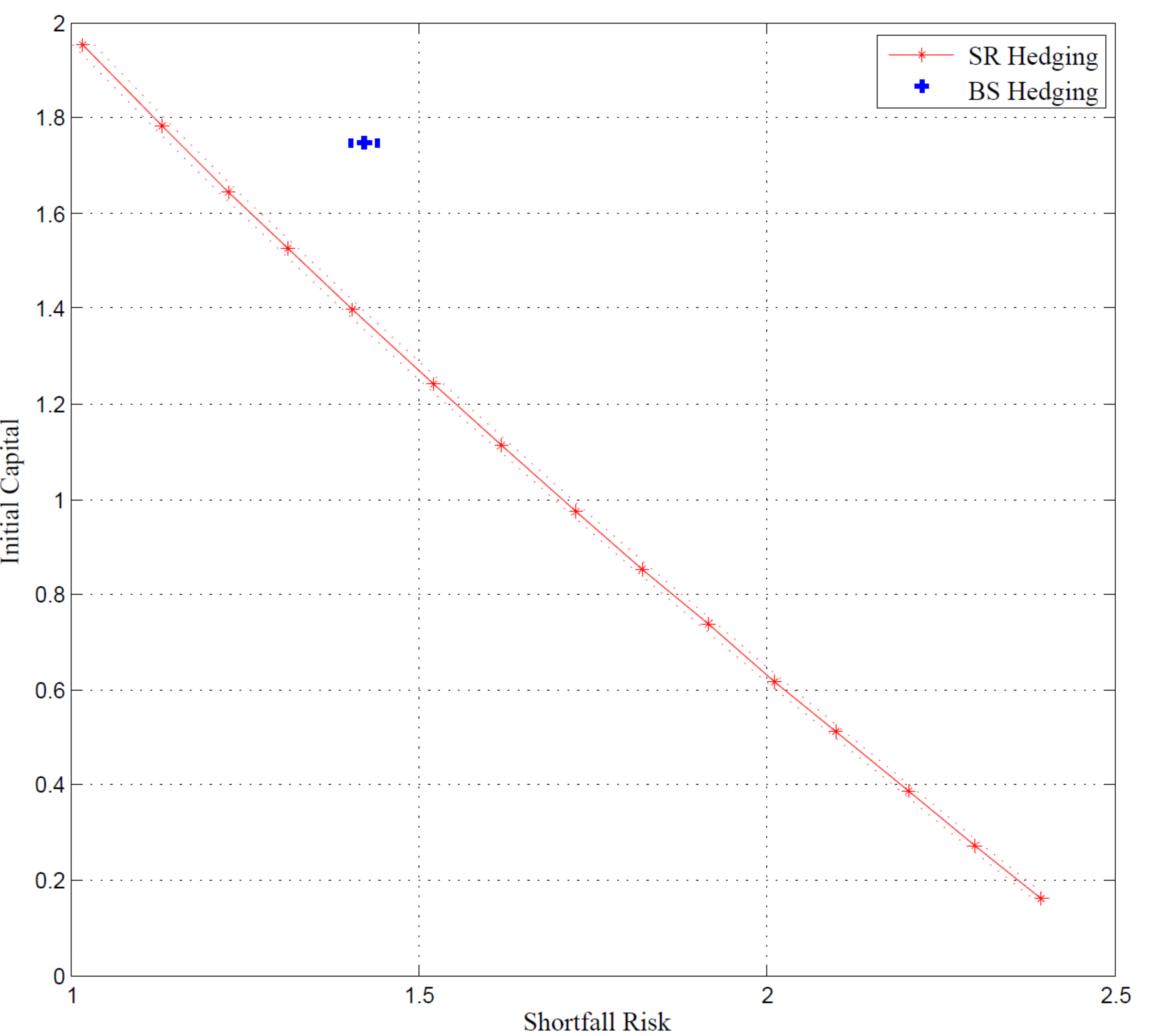}}
\end{center}
\caption{\small Initial capital w.r.t. 
the associated Shortfall Risk of the Black-Scholes strategy (blue)
and the Shortfall strategy (red)
with 95 \% confidence interval (in dotted lines).}
\label{fig:SR_vs_BS}
\end{figure}

\begin{figure}[p]
\begin{center}
\subfigure[$K=0,85 X(t_0)$]
{\includegraphics[width=5.5cm, height=3.2cm]
{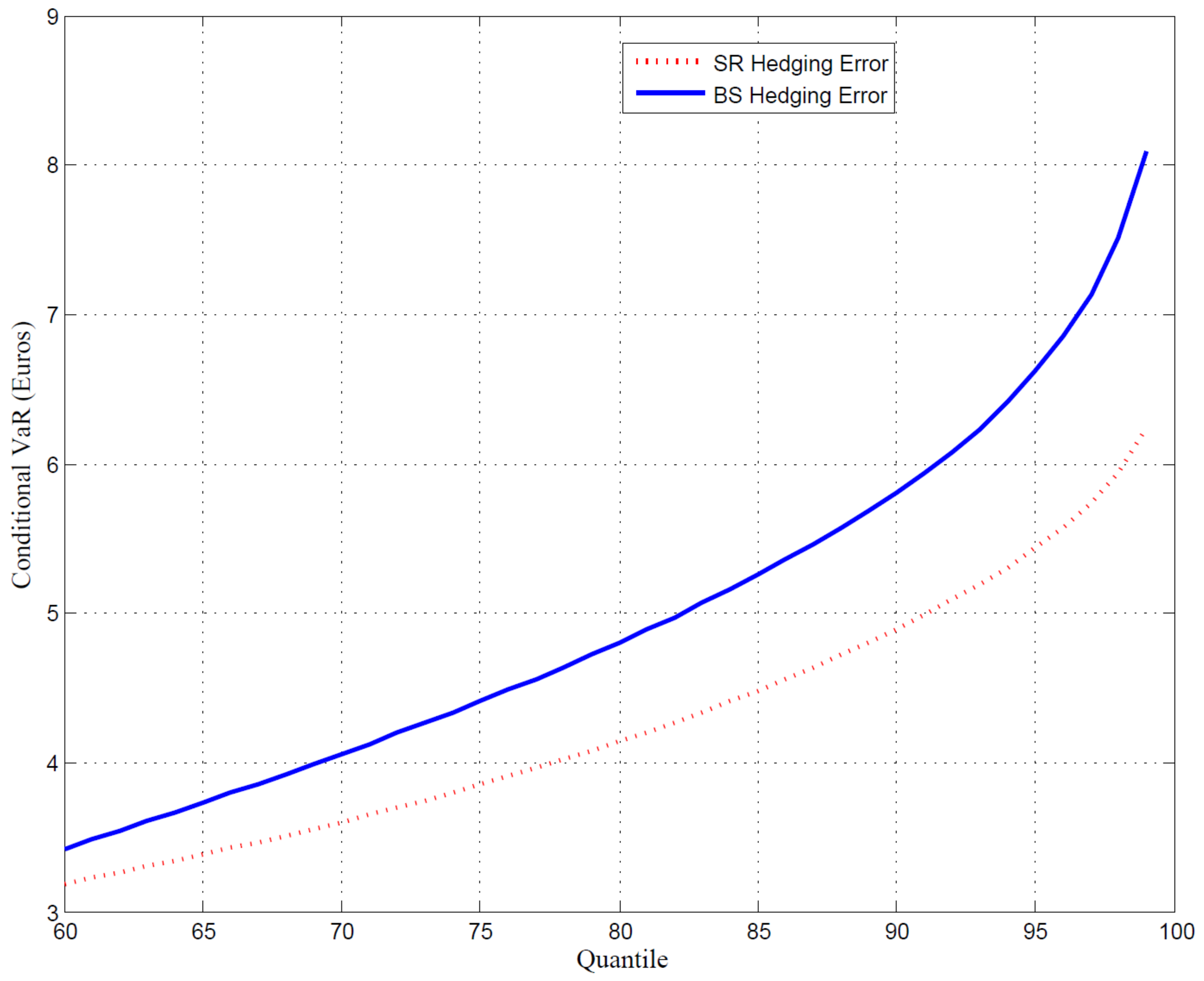}}
\subfigure[$K=0,90 X(t_0)$]
{\includegraphics[width=5.5cm, height=3.2cm]
{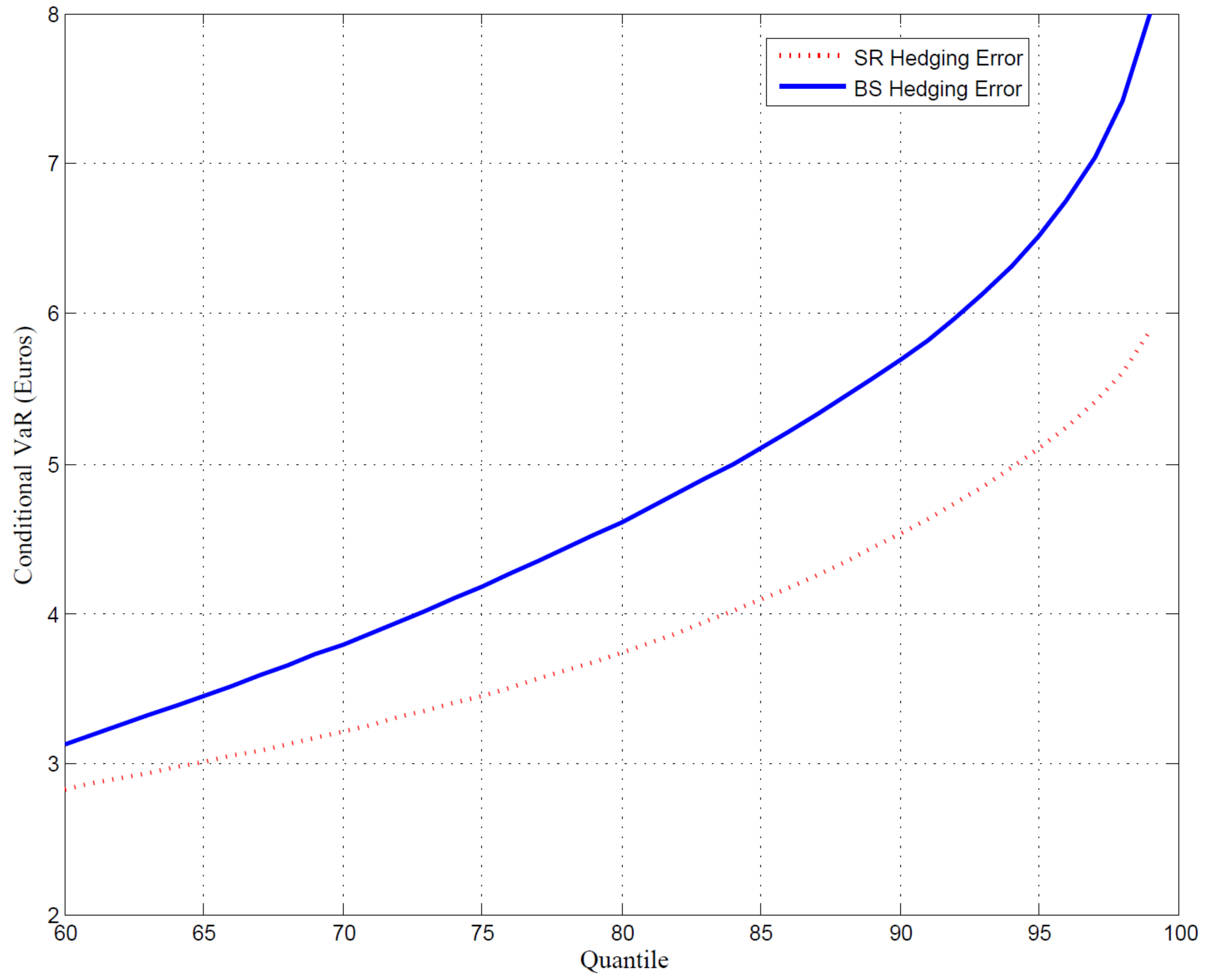}}
\subfigure[$K=0,95 X(t_0)$]
{\includegraphics[width=5.5cm, height=3.2cm]
{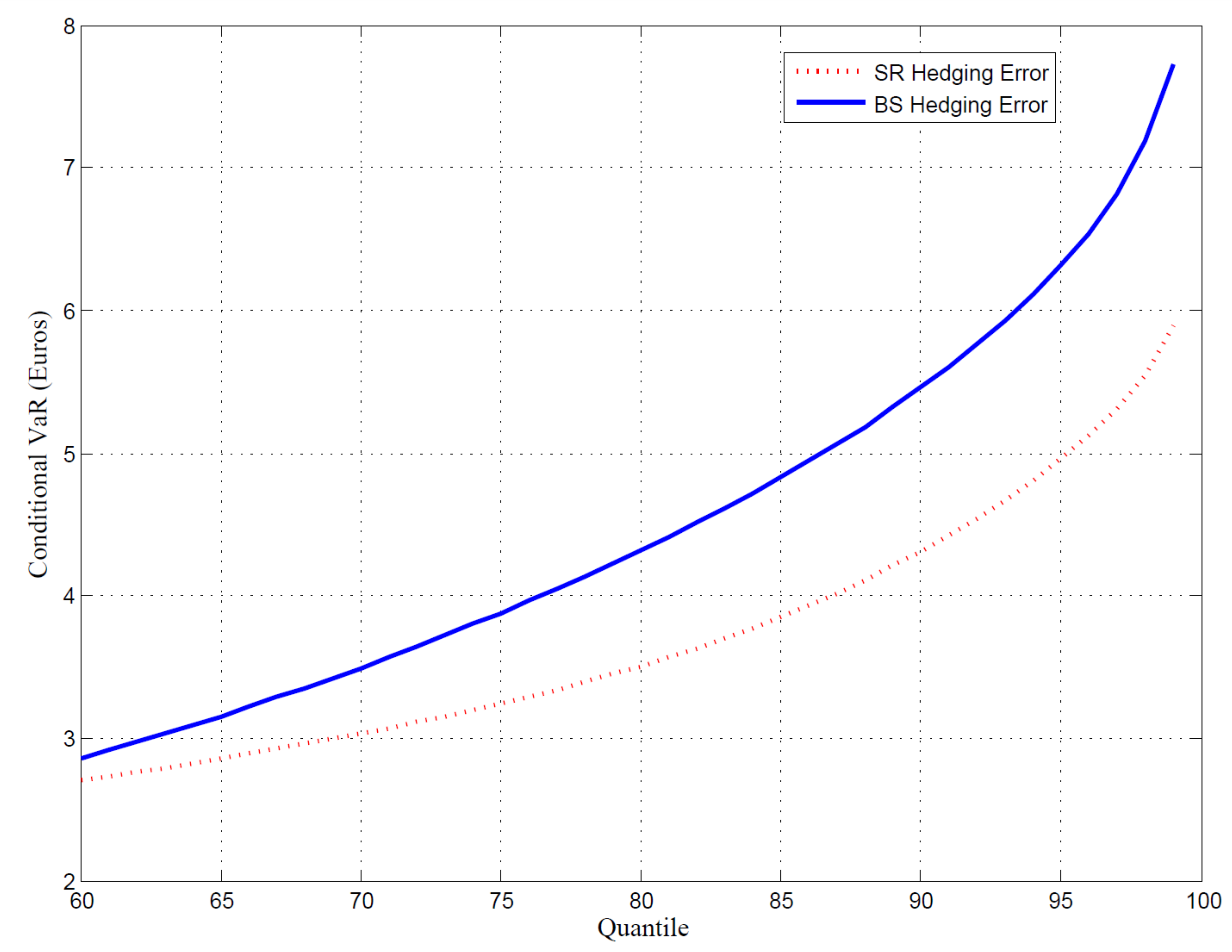}}
\subfigure[$K= X(t_0)$]
{\includegraphics[width=5.5cm, height=3.2cm]
{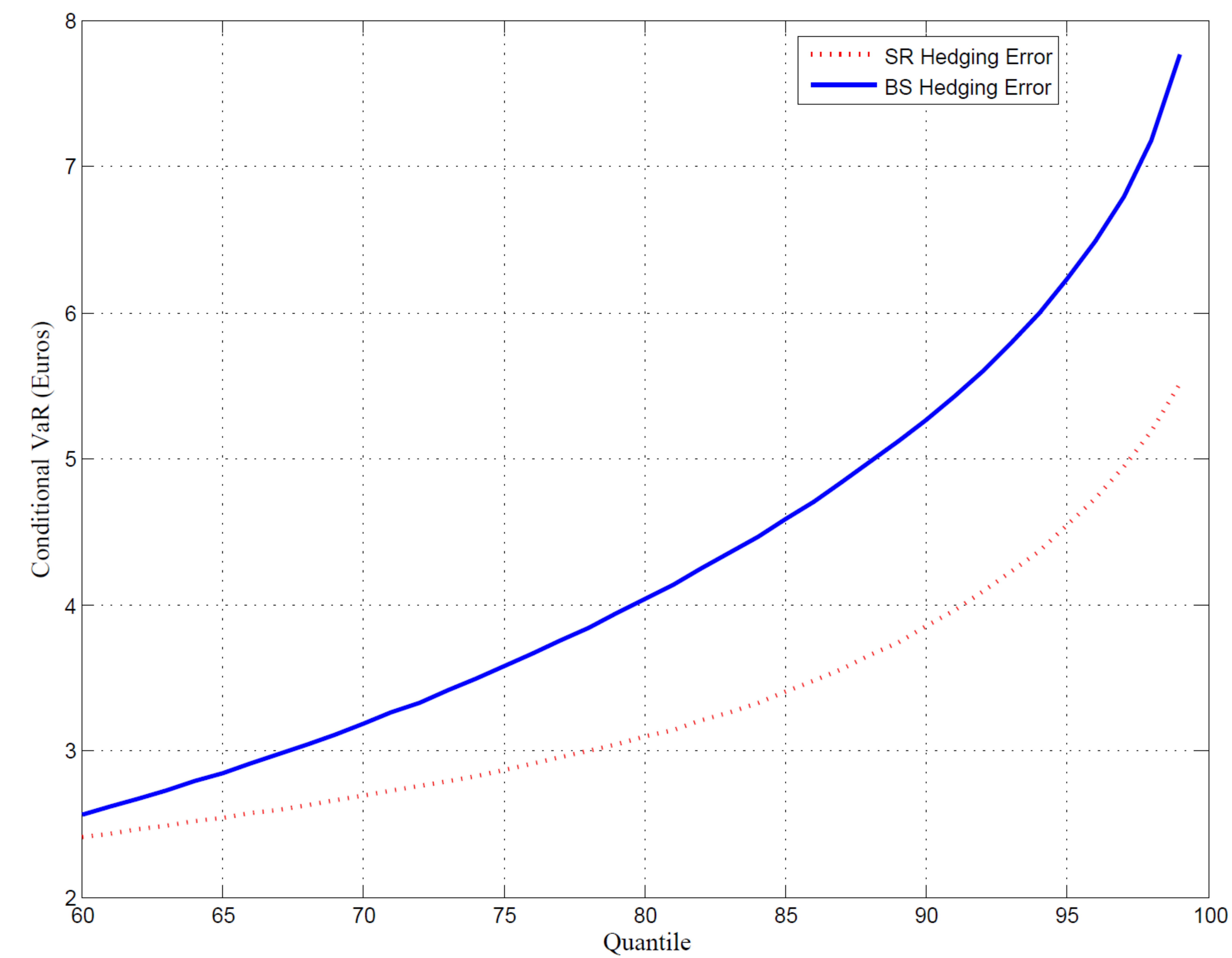}}
\subfigure[$K=1,05 X(t_0)$]
{\includegraphics[width=5.5cm, height=3.2cm]
{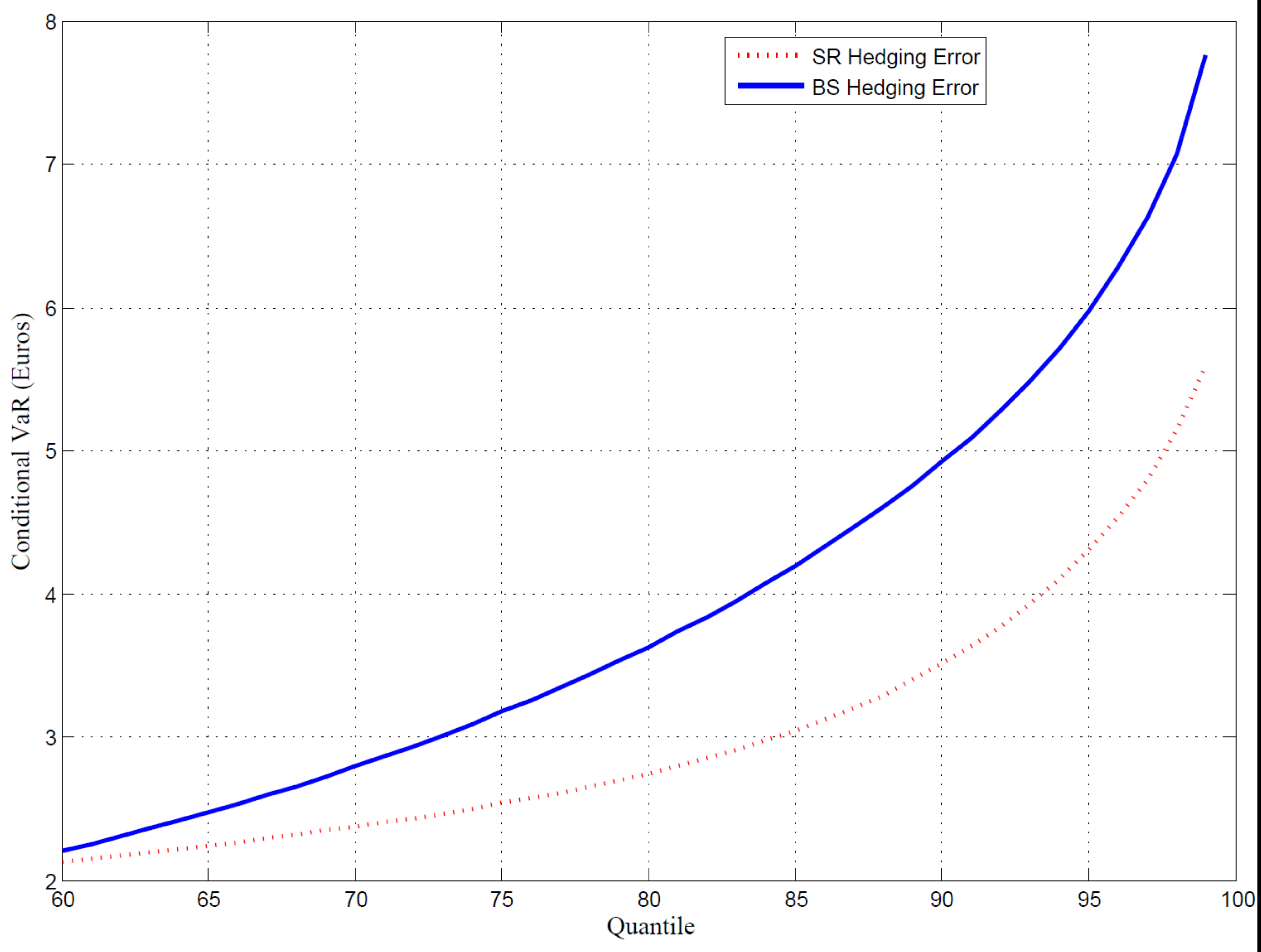}}
\subfigure[$K=1,10 X(t_0)$]
{\includegraphics[width=5.5cm, height=3.2cm]
{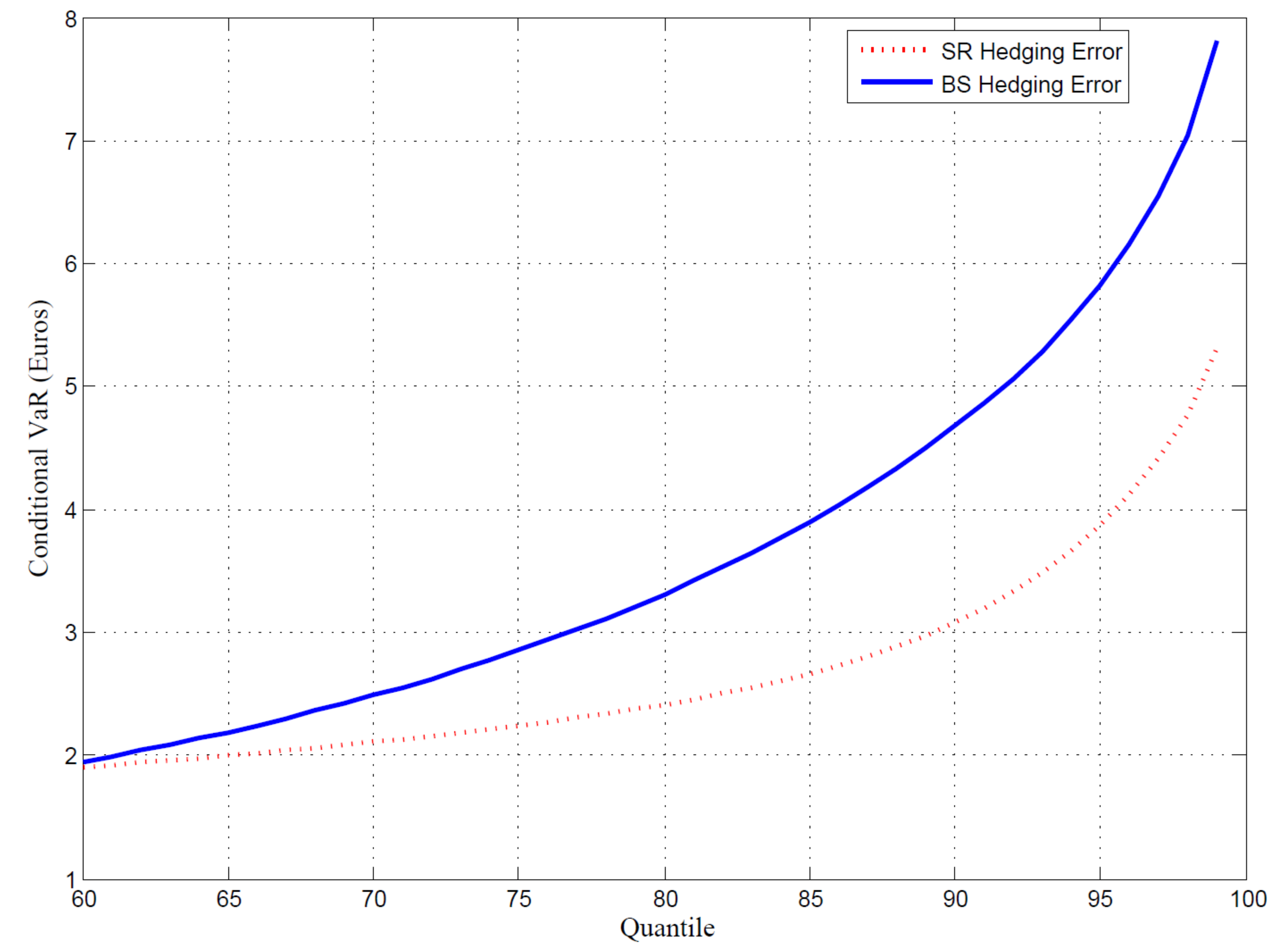}}
\subfigure[$K=1,15 X(t_0)$]
{\includegraphics[width=5.5cm, height=3.2cm]
{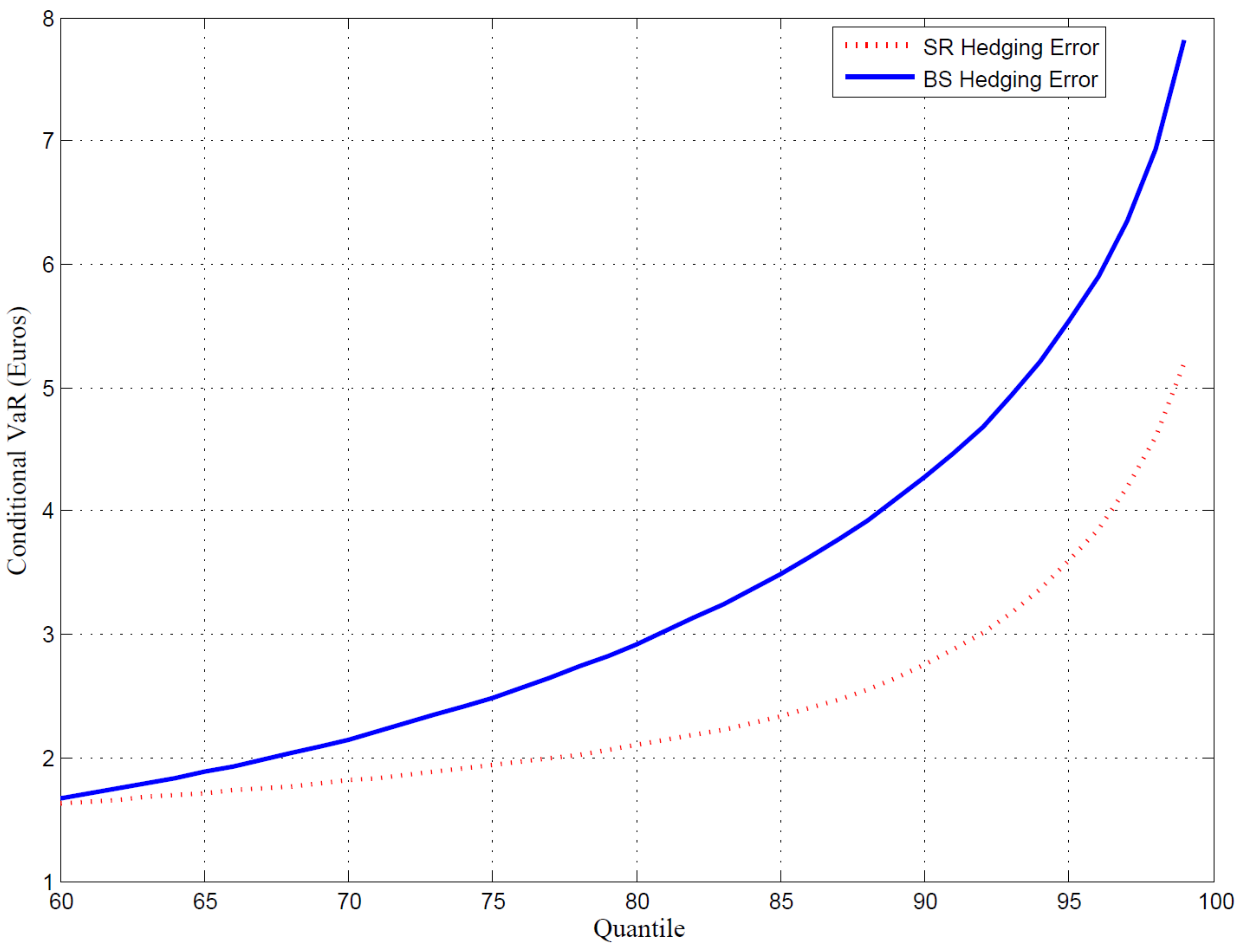}}
\subfigure[$K=1,15 X(t_0)$]
{\includegraphics[width=5.5cm, height=3.2cm]
{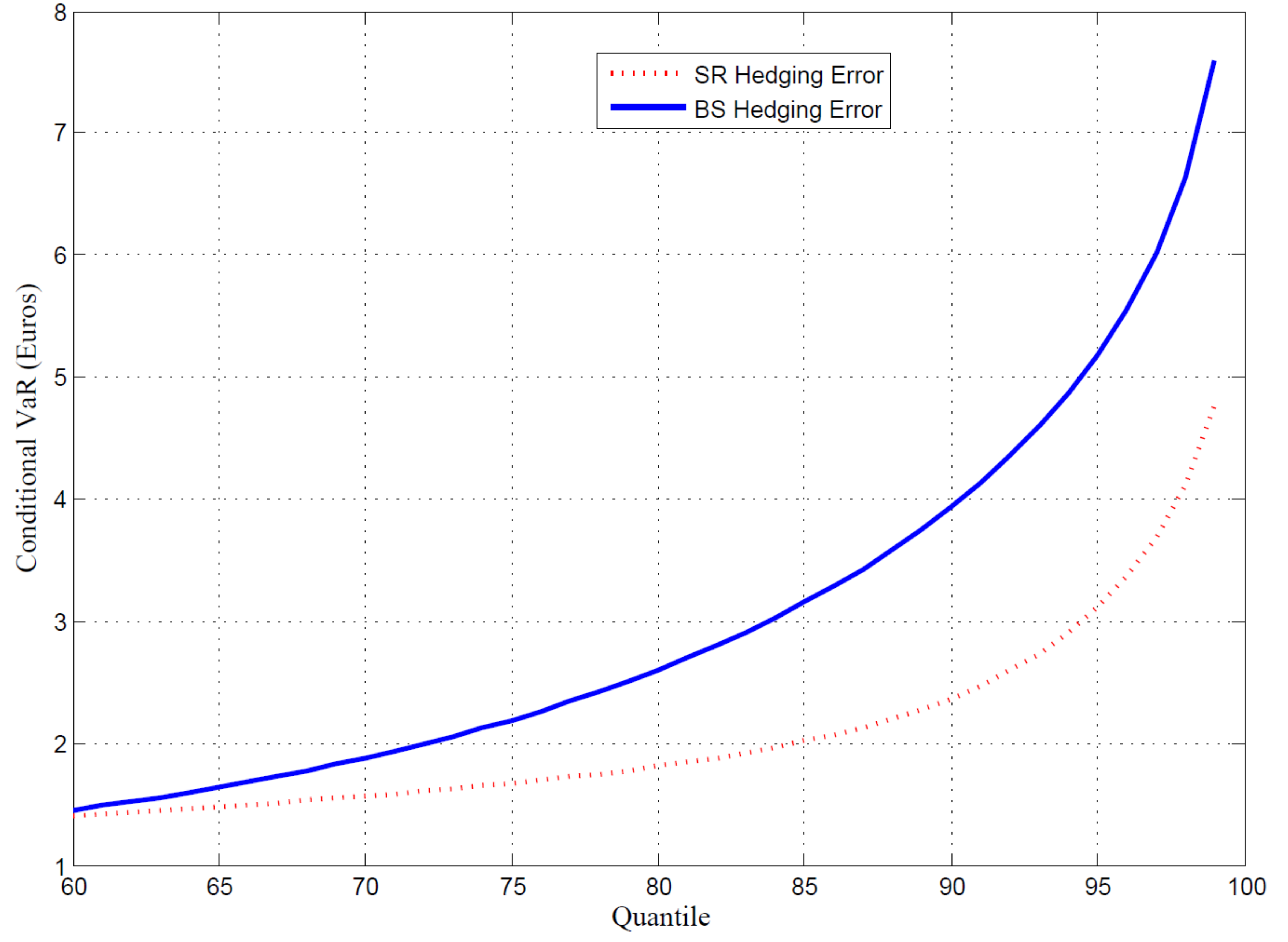}}
\end{center}
\caption{\small CVaR value w.r.t. the quantile level of the the strategies BS (blue)
and SR (red).}
\label{fig:SR_vs_BS_Cvar}
\end{figure}

\section*{Appendix: Geometric Dynamic Principle and HJB equation}

In what follows, 
we put ourselves in the Brownian filtration
setting of \cite{soner2002stochastic} and \cite{bouchard2009stochastic},
which encompasses our framework. 
We omit the presence of $\Lambda$ by assuming $\Pro{\Lambda=1}=1$
and place ourselves on the interval $[0, T^*]$.
We provide one side of the GDP
used to derive the supersolution property.

\begin{Theorem}[Th 3.1, \cite{soner2002stochastic}]
\label{th:GDP1}
Fix $(t,x,p,y)\in \Sb \x [-\kappa, \infty) $ 
such that $y>V(t,x,p)$ and
a family of stopping times 
$\brace{\theta^{\nu, \alpha}~:~ (\nu,\alpha)\in \Uc_{t,y}\x \Ac_{t,p}}$.
Then there exists $(\nu, \alpha)\in \Uc_{t,y} \x\Ac_{t,p}$ 
such that
$$
Y^{t,x,y,\nu}_{\theta^{\nu, \alpha}}\ge 
V(\theta^{\nu, \alpha}, X^{t,x}_{\theta^{\nu, \alpha}}, P_{\theta^{\nu, \alpha}}^{t,p,\alpha}) \ \Pas
$$ 
and 
$Y^{t,x,y,\nu}_{s\wedge \theta^{\nu, \alpha}}\ge -\kappa$ 
for all $s\in [t,T^*]$ $\Pas$
\end{Theorem}

Let $V_*$ be defined by
$
V_* (t,x,p) := \liminf\brace{V(t',x',p') ~:~B\ni(t',x',p')\to (t,x,p) }
$
where $B$ denotes an open subset of $[0,T^*]\x \R^*_+\x \R^*_-$ with
$(t,x,p)\in \cl (B)$.
We assume that 
$V$ is locally bounded on $\Sb$,
so that $V_*$ is finite.
In what follows,
we introduce only the supersolution property for $V_*$,
deriving from Theorem \ref{th:GDP1},
in the special case given by dynamics
\eqref{eq:price dynamics},\eqref{eq: portfolio process}
and \eqref{eq: process P}.

For $\eps\ge0$, 
we introduce the relaxed operator
$\Theta \mapsto \bar{F}_\eps (\Theta)$
for the variable 
$$\Theta = (t,x,p,d,d_x, d_p, d_{xx}, d_{pp}, d_{xp})\in\Sb \x \R^6$$
given by
\begin{equation}
\bar{F}_\eps (\Theta):= \hspace{-0.1cm}
\label{eq: HJB formelle}
\sup\limits_{(u,a)\in \Nc_\eps(\Theta)} \hspace{-0.1cm}
\brace{\left(u - d_x\right)\mu(t, x)
- \demi \left(\sigma^2(t,x)d_{xx} +a^2p^2  d_{pp}+2 ap \sigma(t,x)d_{xp}\right)}
\end{equation}
with
\begin{equation}
\label{eq: N_epsilon}
\Nc_\eps (\Theta):= \brace{ (u, a)\in \R^2  ~:~ 
\left|\sigma(t,x)\left(u- d_x\right) - a p  d_p \right|\le \eps}
\; .
\end{equation}
to finally introduce 
$
\bar{F}^*(\Theta):= \limsup\brace{\bar{F}_\eps (\Theta') ~:~ \eps\searrow 0, \Theta'\rightarrow\Theta}
$.
We adopt the convention $\sup\emptyset = -\infty$ and 
$$
\bar{F}^* \vp = \bar{F}^* (t,x,p, \vp(x,p), \vp_x (t,x),\vp_p (t,x), \vp_{xx}(t,x),\vp_{pp}(t,x),\vp_{xp}(t,x))
$$	
for a smooth function $\vp$.
We hence formulate the supersolution property of 
$V_*$.
For definitions and use of viscosity solutions,
we refer to \cite{crandall1992user}.
The supersolution property
inside the domain is given by 
Theorem 2.1 and Corollary 3.1 in \cite{bouchard2009stochastic}.
The boundary condition at $t=T^*$ is
given by Theorem 2.2 in \cite{bouchard2009stochastic}.
In our case, 
by assuming the concavity of $\Psi$
in $y$,
we have the convexity of $\Psi^{-1}$
in the $p$ variable.
We also have $\Nc_0(\Theta)\ne \emptyset$
for any $\Theta$.
According to these two properties,
the terminal condition takes a much more simple form.
Altogether, we obtain the following.

\begin{Theorem}[Th. 2.1-2.2, \cite{bouchard2009stochastic}]
\label{th 2.1}
The function $V_*$ is a viscosity supersolution of 
\begin{equation}
 \label{eq:viscosity HJB}
 \left\{
 \begin{array}{ll}
  - \vp_t(t,x,p) + \bar{F}^* \vp(t,x,p) = 0 \quad & \text{on }[0,T^*)\x \R^*_+ \x \R^*_-\\
  V_*(T,x,p) \ge \Psi^{-1}(x,p) \quad & \text{on }\R^*_+ \x \R^*_-
 \end{array}
 \right. \;.
\end{equation}
\end{Theorem}

There is a special Cauchy boundary problem for 
$p=0$ we elude here.
In our case, since $\Psi(0)=0$,
the stochastic target problem
reduces to the superhedging problem.
The target must be reached $\P$-almost surely
and we obtain directly the HJB equation
of \cite{soner2002stochastic}.
In our complete market framework,
it straightly provides 
the Black-Scholes equation \eqref{eq:black-scholes equation}
to which $V_*(.,0)$ is also a viscosity solution.


\end{document}